\documentclass[11pt]{article}
\usepackage{amsfonts, amsmath, amssymb, amsthm, appendix, bbm, bm, color, enumitem, forloop, framed, graphicx, hyperref, lscape, mathrsfs, pdfpages, rotating, sectsty, setspace, threeparttable, tocloft,subfig,multirow}
\usepackage[round]{natbib}  

\hypersetup{pdfborder = {0 0 0},colorlinks=true,linkcolor=blue,citecolor=blue}

\newtheorem{theorem}{Theorem}
\newtheorem{corollary}{Corollary}

\newtheorem{lemma}{Lemma}
\newtheorem{assumption}{Assumption}

\numberwithin{equation}{section}

\theoremstyle{definition}
\newtheorem{remarkTmp}{Remark}
\newenvironment{remark}
	{\medskip  \begin{remarkTmp} 	}
	{ 

		\qed 
		\end{remarkTmp} 
	}
	
\allowdisplaybreaks[1]

\setenumerate[1]{label=\bf(\alph*)}
\setenumerate[2]{label=\bf(\roman*)}

	\DeclareMathOperator*{\argmin}{arg\,min}

	\DeclareMathOperator*{\supp}{supp}

	\newcommand{\E}{\mathbb{E}}
	\newcommand{\En}{\mathbb{E}_n}
	\newcommand{\V}{\mathbb{V}}
	\renewcommand{\P}{\mathbb{P}}

	\newcommand{\R}{\mathbb{R}}

	\newcommand{\One}{\mathbbm{1}}

	\newcommand{\vertiii}[1]{{\left\vert\kern-0.25ex\left\vert\kern-0.25ex\left\vert #1  \right\vert\kern-0.25ex\right\vert\kern-0.25ex\right\vert}}
	
    \def\ddefloop#1{\ifx\ddefloop#1\else\ddef{#1}\expandafter\ddefloop\fi}
    \def\ddef#1{\expandafter\def\csname c#1\endcsname{\ensuremath{\mathcal{#1}}}}
    \ddefloop ABCDEFGHIJKLMNOPQRSTUVWXYZ\ddefloop
    \def\ddef#1{\expandafter\def\csname s#1\endcsname{\ensuremath{\mathsf{#1}}}}
    \ddefloop ABCDEFGHIJKLMNOPQRSTUVWXYZ\ddefloop

    \def\argmin{\operatornamewithlimits{arg\,min}}
    \newcommand{\deq}{:=}

    \newcommand{\bz}{\bm{z}}
    \newcommand{\bZ}{\bm{Z}}
    \newcommand{\bw}{\bm{w}}
    \newcommand{\bx}{\bm{x}}
    \newcommand{\bX}{\bm{X}}
    
    \newcommand{\bv}{\bm{v}}

	\DeclareMathOperator*{\supess}{ess\,sup}

\newcommand{\gt}{\tilde{\gamma}}

\begin{document}

\title{
	\vspace{-0.5in}
	Deep Neural Networks for Estimation and Inference\thanks{
		We thank Milica Popovic for outstanding research assistance. Liang gratefully acknowledges support from the George C. Tiao Fellowship. Misra gratefully acknowledges support from the Neubauer Family Foundation. We are thank Guido Imbens, the handling co-editor, and two anonymous reviewers, as well as Alex Belloni, Xiaohong Chen, Denis Chetverikov, Chris Hansen, Whitney Newey, and Andres Santos, for thoughtful comments, suggestions, and discussions that substantially improved the paper.
	}
}
\author{
	Max H. Farrell \and
	Tengyuan Liang \and
	Sanjog Misra \and
	\\
	University of Chicago, Booth School of Business
}
\date{\today}
\maketitle

\begin{abstract} 
We study deep neural networks and their use in semiparametric inference. We establish novel rates of convergence for deep feedforward neural nets. Our new rates are sufficiently fast (in some cases minimax optimal) to allow us to establish valid second-step inference after first-step estimation with deep learning, a result also new to the literature. Our estimation rates and semiparametric inference results handle the current standard architecture: fully connected feedforward neural networks (multi-layer perceptrons), with the now-common rectified linear unit activation function and a depth explicitly diverging with the sample size. We discuss other architectures as well, including fixed-width, very deep networks. We establish nonasymptotic bounds for these deep nets for a general class of nonparametric regression-type loss functions, which includes as special cases least squares, logistic regression, and other generalized linear models. We then apply our theory to develop semiparametric inference, focusing on causal parameters for concreteness, such as treatment effects, expected welfare, and decomposition effects. Inference in many other semiparametric contexts can be readily obtained. We demonstrate the effectiveness of deep learning with a Monte Carlo analysis and an empirical application to direct mail marketing.
\end{abstract}

\bigskip


\setcounter{page}{0}
\thispagestyle{empty}

\doublespacing

\section{Introduction}
	\label{sec:intro}

Statistical machine learning methods are being rapidly integrated into the social and medical sciences. Economics is no exception, and there has been a recent surge of research that applies and explores machine learning methods in the context of econometric modeling, particularly in ``big data'' settings. Furthermore, theoretical properties of these methods are the subject of intense recent study. This has netted several breakthroughs both theoretically, such as robust, valid inference following machine learning, and in novel applications and conclusions. Our goal in the present work is to study a particular statistical machine learning technique which is widely popular in industrial applications, but less frequently used in academic work and largely ignored in recent theoretical developments on inference: deep neural networks. To our knowledge we provide the first inference results using deep learning methods.

Neural networks are estimation methods that model the relationship between inputs and outputs using layers of connected computational units (neurons), patterned after the biological neural networks of brains. These computational units sit between the inputs and output and allow data-driven learning of the appropriate model, in addition to learning the parameters of that model. Put into terms more familiar in nonparametric econometrics: neural networks can be thought of as a (complex) type of sieve estimation where the basis functions are flexibly learned from the data. Neural networks are perhaps not as familiar to economists as other methods, and indeed, were out of favor in the machine learning community for several years, returning to prominence only very recently in the form of deep learning. Deep neural nets contain many hidden layers of neurons between the input and output layers, and have been found to exhibit superior performance across a variety of contexts. Our work aims to bring wider attention to these methods and to take the first step toward filling the gaps in the theoretical understanding of inference using deep neural networks. Our results can be used in many economic contexts, including selection models, games, consumer surplus, and dynamic discrete choice.

Before the recent surge in attention, neural networks had taken a back seat to other methods (such as kernel methods or forests) largely because of their modest empirical performance and challenging optimization. However, the availability of scalable computing and stochastic optimization techniques \citep{lecun1998gradient,kingma2014adam} and the change from smooth sigmoid-type activation functions to rectified linear units (ReLU), $x \mapsto \max(x,0)$ \citep{nair2010rectified}, have seemingly overcome optimization hurdles and now this form of deep learning matches or sets the state of the art in many prediction contexts \citep{krizhevsky2012imagenet, he2016identity}. Our theoretical results speak directly to this modern implementation of deep learning: we explicitly model the depth of the network as diverging with the sample size and focus on the ReLU activation function.

Further back in history, before falling out of favor, neural networks were widely studied and applied, particularly in the 1990s. In that time, \emph{shallow} neural networks with smooth activation functions were shown to have many good theoretical properties. Intuitively, neural networks are a form of sieve estimation, wherein basis functions of the original variables are used to approximate unknown nonparametric objects. What sets neural nets apart is that the basis functions are themselves learned from the data by optimizing over many flexible combinations of simple functions. It has been known for some time that such networks yield universal approximations \citep{hornik1989multilayer}. Comprehensive theoretical treatments are given by \citet{white1992artificial} and \citet{Anthony-Bartlett1999_book}. Of particular relevance in this strand of theoretical work is \citet{Chen-White1999_IEEE}, where it was shown that single-layer, sigmoid-based networks could attain sufficiently fast rates for semiparametric inference (see \citet{Chen2007_handbook} for more references).

We explicitly depart from the extant literature by focusing on the modern setting of deep neural networks with the rectified linear (ReLU) activation function. We provide nonasymptotic bounds for nonparametric estimation using deep neural networks, immediately implying convergence rates.  The bounds and convergence rates appear to be new to the literature and are one of the main theoretical contributions of the paper. We provide results for a general class of smooth loss functions for nonparametric regression style problems, covering as special cases generalized linear models and other empirically useful contexts. In our application to causal inference we specialize our results to linear and logistic regression as concrete illustrations. Our proof strategy employs a localization analysis that uses scale-insensitive measures of complexity, allowing us to consider richer classes of neural networks. This is in contrast to analyses which restrict the networks to have bounded parameters for each unit (discussed more below) and to the application of scale sensitive measures such as metric entropy \citep[used by][for example]{Chen-White1999_IEEE}. These approaches would not deliver our sharp bounds and fast rates. Recent developments in approximation theory and complexity for deep ReLU networks are important building blocks for our results.

Our second main result establishes valid inference on finite-dimensional parameters following first-step estimation using deep learning. We focus on causal inference for concreteness and wide applicability, as well as to allow direct comparison to the literature. Program evaluation with observational data is one of the most common and important inference problems, and has often been used as a test case for theoretical study of inference following machine learning \citep[e.g.,][]{Belloni-Chernozhukov-Hansen2014_REStud,Farrell2015_arXiv,Belloni-etal2017_Ecma,Athey-Imbens-Wager2018_JRSSB}. Causal inference as a whole is a vast literature; see \cite{Imbens-Rubin2015_book} for a broad review and \cite{Abadie-Cattaneo2018_ARE} for a recent review of program evaluation methods, and further references in both. Deep neural networks have been argued (experimentally) to outperform the previous state-of-the-art in causal inference \citep{westreich2010propensity,johansson2016learning,shalit2017estimating,hartford2017deep}. To the best of our knowledge, ours are among the first theoretical results that explicitly deliver inference using deep neural networks.

We give specific results for average treatment effects, counterfactual expected utility/profits from treatment targeting strategies, and decomposition effects. Our results allow planners (e.g., firms or medical providers) to compare different strategies, either predetermined or estimated using auxiliary data, and recognizing that targeting can be costly, decide which strategy to implement. An interesting, and potentially useful, point we make in this context is that the selection on observables framework yields identification of counterfactual average outcomes without additional structural assumptions, so that, e.g., expected profit from a counterfactual treatment rule can be evaluated.

The usefulness of our deep learning results is of course not limited to causal inference. In particular, our results yield inference on essentially any estimand that admits a locally robust estimator \citep{Chernozhukov-etal2018_WP} that depends only on target functions within our class of loss function (under appropriate regularity conditions). Our aim is not to innovate at the semiparametric step, for example by seeking weaker conditions on the first stage, but rather, we aim to utilize such results. Prior work has verified the high-level conditions for other first-stage estimators, such as traditional kernels or series/sieves, lasso methods, sigmoid-based shallow neural networks, and others (under suitable assumptions for each method). Our work contributes directly to this area of research by showing that deep nets are a valid and useful first-step estimator, in particular, attaining a rate of $o(n^{-1/4})$ under appropriate smoothness conditions. Finally, we do not rely on sample splitting or cross fitting. In particular, we use localization explicitly to directly verify conditions required for valid inference, which may be a novel application of this proof method that is useful in future semiparametric inference problems.

We numerically illustrate our results, and more generally the utility of deep learning, with a detailed simulation study and an empirical study of a direct mail marketing campaign. Our data come from a large US consumer products retailer and consists around to three hundred thousand consumers with one hundred fifty covariates. \cite{Hitsch-Misra2018_WP} recently used this data to study various estimators, both traditional and modern, of heterogeneous treatment effects. We refer the reader to that paper for a more complete description of the data as well as results using other estimators (see also \cite{Hansen-Kozbur-Misra2017_WP}). We study the effect of catalog mailings on consumer purchases, and moreover, compare different targeting strategies (i.e.\ to which consumers catalogs should be mailed). The cost of sending out a single catalog can be close to one dollar, and with millions being set out, carefully assessing the targeting strategy is crucial. Our results suggest that deep nets are at least as good as (and sometimes better than) the best methods in \cite{Hitsch-Misra2018_WP}.

The remainder of the paper proceeds as follows. Next, we briefly review the related theoretical literature. Section \ref{sec:deep nets} introduces deep ReLU networks and states our main theoretical results: nonasymptotic bounds and convergence rates for general nonparametric regression-type loss functions. The semiparametric inference problem is set up in Section \ref{sec:params} and asymptotic results are presented in Section \ref{sec:inference}. The empirical application is presented in Section \ref{sec:application}. Results of a simulation study are reported in Section \ref{sec:simuls}. Section \ref{sec:conclusion} concludes. All proofs are given in the appendix.

\subsection{Related Theoretical Literature}
	\label{sec:literature}

Our paper contributes to several rapidly growing literatures, and we can not hope to do justice to each here. We give only those citations of particular relevance; more references can be found within these works. First, there has been much recent study of the statistical properties of the machine learning tools as an end in itself. Many studies have focused on the lasso and its variants \citep{Bickel-Ritov-Tsybakov2009_AoS,Belloni-Chernozhukov-Wang2011_Bmka,BCCH2012_Ecma,Farrell2015_arXiv} and tree/forest based methods \citep{Wager-Athey2018_JASA}, though earlier work studied shallow (typically with a single hidden layer) neural networks with smooth activation functions \citep{white1989learning, white1992artificial, Chen-White1999_IEEE}. We fill the gap in this literature by studying \emph{deep} neural networks with the non-smooth ReLU activation.

A second, intertwined strand of literature focuses on inference following the use of machine learning methods, often with a focus on average causal effects. Initial theoretical results were concerned with obtaining valid inference on a coefficient in a high-dimensional regression, following model selection or regularization, with particular focus on the lasso \citep{BCCH2012_Ecma,Javanmard-Montanari2014_JMLR,vandeGeer-etal2014_AoS}. Intuitively, this is a semiparametric problem, where the coefficient of interest is estimable at the parametric rate and the remaining coefficients are collectively a nonparametric nuisance parameter estimated using machine learning methods. Building on this intuition, many have studied the semiparametric stage directly, such as obtaining novel, weaker conditions easing the application of machine learning methods \citep[][and references therein]{Belloni-Chernozhukov-Hansen2014_REStud,Farrell2015_arXiv,Chernozhukov-etal2018_WP,Belloni-etal2018_handbook}. Conceptually related to this strand are targeted maximum likelihood \citep{vanderLaan-Rose2011_book} and the higher-order influence functions \citep{Robins-etal2008_IMS,Robins-etal2017_AoS}. Our work builds on this work, employing conditions therein, and in particular, verifying them for deep ReLU nets.

Finally, our convergence rates build on, and contribute to, the recent theoretical machine learning literature on deep neural networks. Because of the renaissance in deep learning, a considerable amount of study has been done in recent years. Of particular relevance to us are \cite{Yarotsky2017_NN, Yarotsky2018_WP} and \cite{Bartlett-etal2017_COLT}; a recent textbook treatment, containing numerous other references, is given by \citet {Goodfellow-Bengio-Courville2016_book}.

\section{Deep Neural Networks}
	\label{sec:deep nets}

In this section we will give our main theoretical results: nonasymptotic bounds and associated convergence rates for deep neural network estimation. The utility of these results for second-step semiparametric causal inference (the downstream task), for which our rates are sufficiently rapid, is demonstrated in Section \ref{sec:inference}. We view our results as an initial step in establishing both the estimation and inference theory for modern deep learning, i.e.\ neural networks built using the multi-layer perceptron architecture (described below) and the nonsmooth ReLU activation function. This combination is crucial: it has demonstrated state of the art performance empirically and can be feasibly optimized. This is in contrast with sigmoid-based networks, either shallow (for which theory exists, but may not match empirical performance) or deep (which are not feasible to optimize), and with shallow ReLU networks, which are not known to approximate broad classes functions.

As neural networks are perhaps less familiar to economists and other social scientists, we first briefly review the construction of deep ReLU nets. Our main focus will be on the fully connected feedfoward neural network, frequently referred to as a multi-layer perceptron, as this is the most commonly implemented network architecture and we want our results to inform empirical practice. However, our results are more general, accommodating other architectures provided they are able to yield a universal approximation (in the appropriate function class), and so we review neural nets more generally and give concrete examples.

Our goal is to estimate an unknown, assumed-smooth function $f_*(\bx)$, that relates the covariates $\bX \in \R^d$ to an outcome $Y$ as the minimizer of the expectation of the per-observation loss function. Collecting these random variables into the vector $\bZ = (Y, \bX')' \in \R^{d+1}$, with $\bz = (y, \bx')'$ denoting a realization, we write
\[
f_* = \argmin \E \left[ \ell\left(f, \bZ\right) \right]. 
\]
We allow for any loss function that is Lipschitz in $f$ and obeys a curvature condition around $f_*$. Specifically, for constants $c_1$, $c_2$, and $C_\ell$ that are bounded and bounded away from zero, we assume that $\ell(f,\bz)$ obeys
\begin{align}
	\begin{split}
		\label{eqn:loss}
		& |\ell(f, \bz) - \ell(g, \bz)| \leq C_\ell |f(\bx) - g(\bx)| ,  		\\
		& c_1 \E \left[ (f - f_*)^2 \right]\leq \E [\ell(f, \bZ)] - \E [\ell (f_*, \bZ)] \leq c_2 \E \left[(f - f_*)^2\right].
	\end{split}
\end{align}
Our results will be stated for a general loss obeying these two conditions.\footnote{We thank an anonymous referee for suggesting this approach.} We give a unified localization analysis of all such problems. This family of loss function covers many interesting problems. Two leading examples, used in our application to causal inference, are least squares and logistic regression, corresponding to the outcome and propensity score models respectively. For least squares, the target function and loss are 
\begin{equation}
	\label{eqn:ols}
	f_*(\bx) \deq \E[ Y | \bX = \bx]   		\qquad \text{ and } \qquad  		  \ell\left(f, \bz \right) = \frac{1}{2} (y - f(\bx))^2,
\end{equation}
respectively, while for logistic regression these are
\begin{equation}
	\label{eqn:logit}
	f_*(\bx) \deq\log \frac{\E[Y|\bX = \bx]}{1 - \E[Y|\bX = \bx]}   		\qquad \text{ and } \qquad  		  \ell\left(f, \bz \right) = -y f(\bx) + \log\left( 1 + e^{f(\bx)} \right).
\end{equation}
Lemma \ref{lem:ls and logit} verifies, with explicit constants, that \eqref{eqn:loss} holds for these two. Losses obeying \eqref{eqn:loss} extend beyond these cases to other generalized linear models, such as count models, and can even cover multinomial logistic regression (multiclass classification), as shown in Lemma \ref{lem:glms}.

\subsection{Neural Network Constructions}
	\label{sec:construction}

For any loss, we estimate the target function using a deep ReLU network. We will give a brief outline of their construction here, paying closer attention to the details germane to our theory; complete introductions, and further references, are given by \citet{Anthony-Bartlett1999_book} and \citet{Goodfellow-Bengio-Courville2016_book}. 

The crucial choice is the specific network architecture, or class. In general we will call this $\cF_{\rm DNN}$. From a theoretical point of view, different classes have different complexity and different approximating power. We give results for several concrete examples below.  We will focus on \emph{feedforward neural networks}. An example of a feedforward network is shown in Figure \ref{fig:DNN}. The network consists of $d$ input units, corresponding to the covariates $\bX \in \R^d$, one output unit for the outcome $Y$. Between these are $U$ hidden units, or computational nodes or neurons. These are connected by a directed acyclic graph specifying the architecture. The key graphical feature of a feedforward network is that hidden units are grouped in a sequence of $L$ layers, the \emph{depth} of the network, where a node is in layer $l = 1, 2, \ldots, L$, if it has a predecessor in layer $l-1$ and no predecessor in any layer $l' \geq l$. The \emph{width} of the network at a given layer, denoted $H_l$, is the number of units in that layer. The network is completed with the choice of an \emph{activation function} $\sigma: \R \mapsto \R$ applied to the output of each node as described below. In this paper, we focus on the popular ReLU activation function $\sigma(x) = \max(x, 0)$, though our results can be extended (at notational cost) to cover piecewise linear activation functions (see also Remark \ref{rem:activation}).

\begin{figure}
	\centering
	\includegraphics[width=0.45\textwidth]{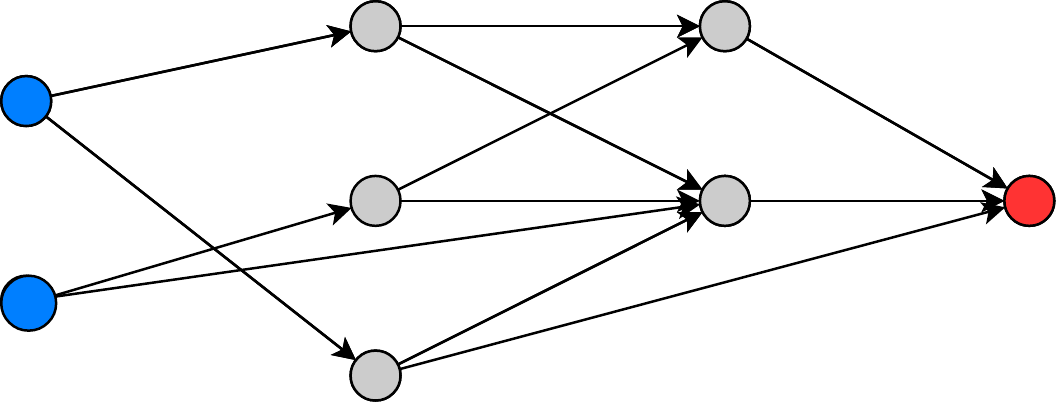}
	\caption{Illustration of a feedforward neural network with $W=18$, $L=2$, $U=5$, and input dimension $d=2$. The input units are shown in blue at left, the output in red at right, and the hidden units in grey between them.}
	\label{fig:DNN}
\end{figure}

An important and widely used subclass is the one that is \textit{fully connected} between consecutive layers but has \emph{no} other connections and each layer has number of hidden units that are of the same order of magnitude. This architecture is often referred to as a \textit{Multi-Layer Perceptron} (MLP) and we denote this class as $\cF_{\rm MLP}$. See Figure \ref{fig:DNN-MLP}, cf.\ Figure \ref{fig:DNN}. We will assume that all the width of all layers share a common asymptotic order $H$, implying that for this class $U \asymp L H$.

\begin{figure}
  \centering
	\includegraphics[width=0.45\textwidth]{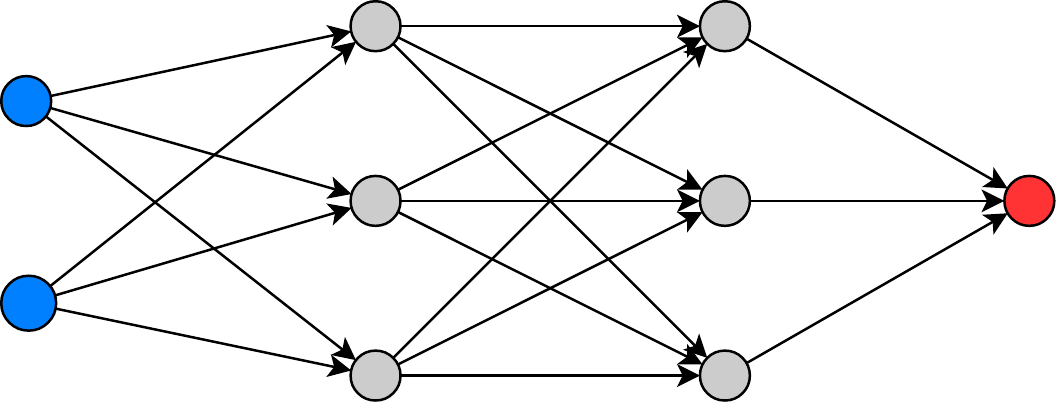}
	\caption{Illustration of multi-layer perceptron $\cF_{\rm MLP}$ with $H=3$, $L=2$ ($U = 6$, $W=25$), and input dimension $d=2$.}
	\label{fig:DNN-MLP}
\end{figure}

We will allow for generic feedforward networks in our results, but we present special results for the MLP case, as it is widely used in empirical practice. As we will see below, the architecture, through its complexity, and more importantly, approximation power, plays a crucial role in the final convergence rate. In particular, we find only a suboptimal rate for the MLP case, but our upper bound is still sufficient for semiparametric inference. As a note on exposition, while our main results are in fact nonasymptotic bounds that hold with high probability, for simplicity we will refer to them as ``rates'' in most discussion.

To build intuition on the computation, and compare to other nonparametric methods, let us focus on least squares for the moment, i.e.\ Equation \eqref{eqn:ols}, with a continuous outcome using a multilayer perceptron with constant width $H$. Each hidden unit $u$ receives an input in the form of a linear combination $\tilde{\bx}'\bw + b$, and then returns $\sigma(\tilde{\bx}'\bw + b)$, where the vector $\tilde{\bx}$ collects the output of all the units with a directed edge into $u$ (i.e., from prior layers), $\bw$ is a vector of weights, and $b$ is a constant term. (The constant term is often referred to as the ``bias'' in the deep learning literature, but given the loaded meaning of this term in inference, we will largely avoid referring to $b$ as a bias.) The final layer's output is simply $\tilde{\bx}'\bw + b$ in the least squares case. The collection, over all nodes, of $\bw$ and $b$, constitutes the parameters $\theta$ which are optimized in the final estimation. We denote $W$ as the total number of parameters of the network. For the MLP, $W = (d+1) H + (L-1) (H^2 + H) + H + 1$. In general, $W$, $U$, $L$, and $H$, may change with $n$, but we suppress this in the notation.

Optimization proceeds layer-by-layer using (variants of) stochastic gradient descent, with gradients of the parameters calculated by back-propagation (implementing the chain rule) induced by the network structure. To see this, let $\tilde{x}_{h,l}$ denote the scalar output of a node $u = (h,l)$, for $h = 1, \ldots H$, $l = 1, \ldots L$, and let $\tilde{\bx}_{l} = (\tilde{x}_{1,l}, \ldots, \tilde{x}_{H,l})'$ for layer $l \leq L$. Each node thus computes $\tilde{x}_{h,l} = \sigma(\tilde{\bx}_{l-1}'\bw_{h,l-1} + b_{h,l-1})$ and the final output is $\hat{y} = \tilde{\bx}_{L}'\bw_{L} + b_{L}$. Once we recall that the network begins with the original observation $\bx$, we can view $\tilde{\bx}_{L} = \tilde{\bx}_{L}(\bx)$, and thus the final output may be seen as a basis function approximation (albeit a complex and random one) written as $\hat{f}_{\rm MLP}(\bx) = \tilde{\bx}_{L}(\bx)'\bw_{L} + b_{L}$, which is reminiscent of a traditional series (linear sieve) estimator. If all layers save the last were fixed, we could simply optimize using least squares directly: $(\bw_{L}, b_{L}) = \argmin_{\bw, b} \| y_i - \tilde{\bx}_{L}'\bw - b\|_n^2$.

The crucial distinction is that the basis functions $\tilde{\bx}_{L}(\cdot)$ are learned from the data. The ``basis'' is $\tilde{\bx}_{L}  = (\tilde{x}_{1,L}, \ldots, \tilde{x}_{H,L})'$, where each $\tilde{x}_{h,L} = \sigma(\tilde{\bx}_{L-1}'\bw_{h,L-1} + b_{h,L-1})$. Therefore, ``before'' we can solve the least squares problem above, we would have to estimate $(\bw_{h,L-1}', b_{h,L-1}), h = 1, \ldots, H$, anticipating the final estimation. These in turn depend on the prior layer, and so forth back to the original inputs $\bX$. Measuring the gradient of the loss with respect to each layer of parameters uses the chain rule recursively, and is implemented by back-propagation. This is simply a sketch of course; for further introduction, see \citet{Hastie-Tibshirani-Friedman2009_book} and \cite{Goodfellow-Bengio-Courville2016_book}.

To further clarify the use of deep nets, it is useful to make explicit analogies to more classical nonparametric techniques, leveraging the form $\hat{f}_{\rm MLP}(\bx) = \tilde{\bx}_{L}(\bx)'\bw_{L} + b_{L}$. For a traditional series estimator, say smoothing splines, the two choices for the practitioner are the spline basis (the shape and the degree) and the number of terms (knots), commonly referred to as the smoothing and tuning parameters, respectively. In kernel regression, these would respectively be the shape of the kernel (and degree of local polynomial) and the bandwidth(s). For neural networks, the same phenomena are present: the architecture as a \emph{whole} (the graph structure and activation function) are the smoothing parameters while the width and depth play the role of tuning parameters for a set architecture.

The architecture plays a crucial role in that it determines the approximation power of the network, and it is worth noting that because of the relative complexity of neural networks, such approximations, and comparisons across architectures, are not simple. It is comparatively obvious that quartic splines are more flexible than cubic splines (for the same number of knots) as is a higher degree local polynomial (for the same bandwidth). At a glance, it may not be clear what function class a given network architecture (width, depth, graph structure, and activation function) can approximate. As we will show below, the MLP architecture is not yet known to yield an optimal approximation (for a given width and depth) and therefore we are only able to prove a bound with slower than optimal rate. As a final note, computational considerations are important for deep nets in a way that is not true conventionally; see Remarks \ref{rem:regularize}, \ref{rem:computation}, and \ref{rem:activation}.

Just as for classical nonparametrics, for a fixed architecture, it is the tuning parameter choices that determine the rate of convergence (for a fixed smoothness of the underlying function). The recent wave of theoretical study of deep learning is still in its infancy. As such, there is no understanding yet of optimal architecture(s) or tuning parameters. Choices of both are quite difficult, and only preliminary research has been done \citep[see e.g.,][and references therein]{daniely2017depth, telgarsky2016benefits, safran2016depth, Mhaskar-Poggio2016_WP,Raghu-etal2017_ICML}. Further exploration of these ideas is beyond the current scope. It is interesting to note that in some cases, a good approximation can be obtained even with a fixed width $H$, provided the network is deep enough, a very particular way of enriching the ``sieve space'' $\cF_{\rm DNN}$; see Corollary \ref{thm:fixed width}.

In sum, for a user-chosen architecture $\cF_{\rm DNN}$, encompassing the choices $\sigma(\cdot)$, $U$, $L$, $W$, and the graph structure, the final estimate is computed using observed samples $\bz_i = (y_i, \bx_i')'$, $i=1,2, \ldots, n$, of $\bZ$, by solving
\begin{align}
	\label{eqn:erm}
	\widehat{f}_{\rm DNN} \deq \argmin_{\substack{f_\theta \in \cF_{\rm DNN} \\ \| f_\theta \|_\infty \leq 2M}}  \sum_{i=1}^n \ell\left(f, \bz_i\right).
\end{align}
Recall that $\theta$ collects, over all nodes, the weights and constants $\bw$ and $b$. When \eqref{eqn:erm} is restricted to the MLP class we denote the resulting estimator $\widehat{f}_{\rm MLP}$. The choice of $M$ may be arbitrarily large, and is part of the definition of the class $\cF_{\rm DNN}$. This is neither a tuning parameter nor regularization in the usual sense: it is not assumed to vary with $n$, and beyond being finite and bounding $\|f_*\|_\infty$ (see Assumption \ref{asmpt:dgp}), no properties of $M$ are required. This is simply a formalization of the requirement that the optimizer is not allowed to diverge on the function level in the $l_\infty$ sense-- the weakest form of constraint. It is important to note that while typically regularization will alter the approximation power of the class, that is not the case with the choice of $M$ as we will assume that the true function $f_*(\bx)$ is bounded, as is standard in nonparametric analysis. With some extra notational burden, one can make the dependence of the bound on $M$ explicit, though we omit this for clarity as it is not related to statistical issues.


\begin{remark}
	\label{rem:regularize}
	In applications it is common to apply some form of regularization to the optimization of \eqref{eqn:erm}. However, in theory, the role of explicit regularization is unclear and may be unnecessary, as stochastic gradient descent presents good, if not better, solutions empirically \citep[see Section \ref{sec:simuls} and][]{zhang2016understanding}. Regularization may improve empirical performance in low signal-to-noise ratio problems. A detailed investigation is beyond the scope of the current work, though we do investigate this numerically in Sections \ref{sec:application} and \ref{sec:simuls}. There are many alternative regularization methods, including $L_1$ and $L_2$ (weight decay) penalties, drop out, and others. 
\end{remark}

\subsection{Bounds and Convergence Rates for Multi-Layer Perceptrons}
	\label{sec:mlp rates}

We can now state our main theoretical results: bounds and convergence rates for deep ReLU networks. All proofs appear in the Appendix. We study neural networks from a nonparametric point of view \citep[e.g.,][in specific scenarios]{white1989learning, white1992artificial,schmidt2017nonparametric,Liang2018_GAN,bauer2017deep}. \cite{Chen-Shen1998_Ecma} and \cite{Chen-White1999_IEEE} share our goal, fast convergence rates for use in semiparametric inference, but focus on shallow, sigmoid-based networks compared to our deep, ReLU-based networks, though they consider dependent data which we do not. Our theoretical approach is quite different. In particular, \cite{Chen-White1999_IEEE} obtain sufficiently fast rates by following the approach of \citet{Barron1993_IEEE} in using Maurey's method \citep{pisier1981remarques} for approximation, but applying the refinement of \cite{makovoz1996random}. Our analysis of deep nets instead employs localization methods \citep{koltchinskii2000rademacher, bartlett2005local, koltchinskii2006local, Koltchinskii2011_book, liang2015learning}, along with the recent approximation work of \citet{Yarotsky2017_NN, Yarotsky2018_WP} and complexity results of \citet{Bartlett-etal2017_COLT}.

The regularity conditions we require are collected in the following. 
\begin{assumption}
	\label{asmpt:dgp}
	Assume that $\bz_i = (y_i, \bx_i')', 1\leq i \leq n$ are i.i.d. copies of $\bZ = (Y, \bX) \in \mathcal{Y} \times [-1,1]^d$, where $X$ is continuously distributed. For an absolute constant $M>0$, assume $\| f_* \|_\infty \leq M$ and $\mathcal{Y} \subset [-M,M]$.
\end{assumption}

This assumption is fairly standard in nonparametrics. The only restriction worth mentioning is that the outcome is bounded. In many cases this holds by default (such as logistic regression, where $\mathcal{Y} = \{0,1\}$) or count models (where $\mathcal{Y} = \{0,1,\ldots,M\}$, with $M$ limited by real-world constraints). For continuous outcomes, such as least squares regression, our restriction is not substantially more limiting than the usual assumption of a model such as $Y = f_*(\bX) + \varepsilon$, where $\bX$ is compact-supported, $f_*$ is bounded, and the stochastic error $\varepsilon$ possesses many moments. Indeed, in many applications such a structure is only coherent with bounded outcomes, such as the common practice of including lagged outcomes as predictors. Next, the assumption of continuously distributed covariates is quite standard. From a theoretical point of view, covariates taking on only a few values can be conditioned on and then averaged over, and these will, as usual, not enter into the dimensionality which curses the rates. Discrete covariates taking on many values may be more realistically thought of as continuous, and it may be more accurate to allow these to slow the convergence rates. Our focus on $L_2(\bX)$ convergence allows for these essentially automatically. Finally, from a practical point of view, deep networks handle discrete covariates seamlessly and have demonstrated excellent empirical performance, which is in contrast to other more classical nonparametric techniques that may require manual adaptation.

Proceeding now to our results, we begin with the most important network architecture, the multi-layer perceptron. This is the most widely used network architecture in practice and an important contribution of our work is to cover this directly, along with ReLU activation. MLPs are now known to approximate smooth functions well, leading to our next assumption: that the target function $f_*$ lies in a Sobolev ball with certain smoothness. Discussion of Sobolev spaces, and comparisons to H\"older and Besov spaces, can be found in \cite{Gine-Nickl2016_book}.  

\begin{assumption}
	\label{asmpt:sobolev}
	Assume $f_*$ lies in the Sobolev ball $\mathcal{W}^{\beta, \infty}([-1, 1]^d)$, with smoothness $\beta \in \mathbb{N}_+$, 
	\begin{equation*}
		f_*(x) \in \mathcal{W}^{\beta, \infty}([-1, 1]^d) := \left\{ f:  \max_{{\bf\alpha}, |{\bf\alpha}| \leq \beta} \supess_{x \in [-1, 1]^d}  |D^{\alpha} f(x)| \leq 1  \right\},
	\end{equation*}
	where ${\bf\alpha} = (\alpha_1, \ldots, \alpha_d)$, $|{\bf\alpha}| = \alpha_1 + \ldots + \alpha_d$ and $D^{{\bf\alpha}} f$ is the weak derivative.
\end{assumption}

Under Assumptions \ref{asmpt:dgp} and \ref{asmpt:sobolev} we obtain the following result, which, to the best of our knowledge, is new to the literature. In some sense, this is our main result for deep learning, as it deals with the most common architecture. We apply this in Sections \ref{sec:inference} and \ref{sec:application} for semiparametric inference. 
\begin{theorem}[Multi-Layer Perceptron]
	\label{thm:mlp rate}
	Suppose Assumptions \ref{asmpt:dgp} and \ref{asmpt:sobolev} hold. Let $\widehat{f}_{\rm MLP}$ be the deep MLP-ReLU network estimator defined by \eqref{eqn:erm}, restricted to $\cF_{\rm MLP}$, for a loss function obeying \eqref{eqn:loss}, with width $H \asymp n^{\frac{d}{2(\beta + d)}} \log^2 n$ and depth $L \asymp \log n$. Then with probability at least $1 - \exp(-n^{\frac{d}{\beta+d}} \log^8 n)$, for $n$ large enough,
	\begin{enumerate}
		\item $\displaystyle \| \widehat{f}_{\rm MLP} - f_* \|_{L_2(\bx)}^2 \leq C \cdot \left\{ n^{-\frac{\beta}{\beta+d}} \log^8 n + \frac{\log \log n}{n} \right\}$ \ and 
		\item $\displaystyle \E_n \left[ (\widehat{f}_{\rm MLP} - f_* )^2 \right] \leq  C \cdot \left\{ n^{-\frac{\beta}{\beta+d}} \log^8 n + \frac{\log \log n}{n} \right\} $,
	\end{enumerate}
	for a constant $C>0$ independent of $n$, which may depend on $d$, $M$, and other fixed constants.
\end{theorem}

Several aspects of this result warrant discussion. We build on the recent results of \citet{Bartlett-etal2017_COLT}, who find nearly-tight bounds on the Vapnik-Chervonenkis (VC) and Pseudo-dimension of deep nets. One contribution of our proof is to use a \emph{scale sensitive} localization theory with \emph{scale insensitive} measures, such as VC- or Pseudo-dimension, for deep neural networks for general smooth loss functions. For the special case of least squares regression, \cite{Koltchinskii2011_book} uses a similar approach, and a similar result to our Theorem \ref{thm:mlp rate}(a) can be derived for this case using his Theorem 5.2 and Example 3 (p.\ 85f).

This approach has two tangible benefits. First, we do not restrict the class of network architectures to have bounded weights for each unit (scale insensitive), in accordance to standard practice \citep{zhang2016understanding} and in contrast to the classic sieve analysis with scale sensitive measure such as metric entropy. Moreover, this allows for a richer set of approximating possibilities, in particular allowing more flexibility in seeking architectures with specific properties, as we explore in the next subsection.  Second, from a technical point of view, we are able to attain a faster rate on the second term of the bound, order $n^{-1}$ in the sample size, instead of the $n^{-1/2}$ that would result from a direct application of uniform deviation bounds. This upper bound informs the trade offs between width and depth, and the approximation power, and may point toward optimal architectures for statistical inference.

This result gives a nonasymptotic bound that holds with high probability. As mentioned above, we will generally refer to our results simply as ``rates'' when this causes no confusion. This result relies on choosing $H$ appropriately given the smoothness $\beta$ of Assumption \ref{asmpt:sobolev}. Of course, the true smoothness is unknown and thus in practice the ``$\beta$'' appearing in $H$, and consequently in the convergence rates, need not match that of Assumption \ref{asmpt:sobolev}. In general, the rate will depend on the smaller of the two. Most commonly it is assumed that the user-chosen $\beta$ is fixed and that the truth is smoother; witness the ubiquity of cubic splines and local linear regression. Rather than spell out these consequences directly, we will tacitly assume the true smoothness is not less than the $\beta$ appearing in $H$ (here and below). Adaptive approaches, as in classical nonparametrics, may also be possible with deep nets, but are beyond the scope of this study.

Even with these choices of $H$ and $L$, the bound of Theorem \ref{thm:mlp rate} is not optimal (for fixed $\beta$, in the sense of \cite{stone1982optimal}). We rely on the explicit approximating constructions of \citet{Yarotsky2017_NN}, and it is possible that in the future improved approximation properties of MLPs will be found, allowing for a sharpening of the results of Theorem \ref{thm:mlp rate} immediately, i.e.\ without change to our theoretical argument. At present, it is not clear if this rate can be improved, but it is sufficiently fast for valid inference.

\subsection{Other Network Architectures}
	\label{sec:other rates}

Theorem \ref{thm:mlp rate} covers only one specific architecture, albeit the most important one at present. However, given that this field is rapidly evolving, it is important to consider other possible architectures which may be beneficial in some cases. To this end, we will state a more generic result and then two specific examples: one to obtain a faster rate of convergence and one for fixed-width networks. All of these results are, at present, more of theoretical interest than practical value, as they are either agnostic about the network (thus infeasible) or rely on more limiting assumptions.

In order to be agnostic about the specific architecture of the network we need to be flexible in the approximation power of the class. To this end, we will replace Assumption \ref{asmpt:sobolev} with the following generic assumption, rather more of a definition, regarding the approximation power of the network. 
\begin{assumption}
	\label{asmpt:bias}
	Let $f_*$ lie in a class $\cF$. For the feedforward network class $\cF_{\rm DNN}$, used in \eqref{eqn:erm}, let the approximation error $\epsilon_{\rm DNN}$ be
	\begin{equation*}
		\epsilon_{\rm DNN} := \sup_{f_* \in \cF} \inf_{\substack{f \in \cF_{\rm DNN} \\ \| f\|_\infty \leq 2M}} \| f -  f_* \|_\infty \enspace.
	\end{equation*}
\end{assumption}
It may be possible to require only an approximation in the $L_2(\bX)$ norm, but this assumption matches the current approximation theory literature and is more comparable with other work in nonparametrics, and thus we maintain the uniform definition. 

Under this condition we obtain the following generic result.

\begin{theorem}[General Feedforward Architecture]
	\label{thm:generic rates}
	Suppose Assumptions \ref{asmpt:dgp} and \ref{asmpt:bias} hold. Let $\widehat{f}_{\rm DNN}$ be the deep ReLU network estimator defined by \eqref{eqn:erm}, for a loss function obeying \eqref{eqn:loss}. Then with probability at least $1 - e^{-\gamma}$, for $n$ large enough,
	\begin{enumerate}
		\item $\displaystyle \| \widehat{f}_{\rm DNN} - f_* \|_{L_2(\bx)}^2  \leq  C \left(  \frac{W L \log W}{n}\log n + \frac{\log \log n + \gamma}{n} + \epsilon_{\rm DNN}^2  \right)$ \ and 
		\item $\displaystyle\E_n \left[ ( \widehat{f}_{\rm DNN} - f_* )^2 \right] \leq C \left(  \frac{W L \log W}{n}\log n + \frac{\log \log n + \gamma}{n} + \epsilon_{\rm DNN}^2  \right) $,
	\end{enumerate}
	for a constant $C>0$ independent of $n$, which may depend on $d$, $M$, and other fixed constants.
\end{theorem}

This is a more general than Theorem \ref{thm:mlp rate}, covering the general deep ReLU network problem defined in \eqref{eqn:erm}, general feedforward architectures, and the general class of losses defined by \eqref{eqn:loss}. The same comments as were made following Theorem \ref{thm:mlp rate} apply here as well: the same localization argument is used with the same benefits. We explicitly use this in the next two corollaries, where we exploit the allowed flexibility in controlling $\epsilon_{\rm DNN}$ by stating results for particular architectures. The bound here is not directly applicable without specifying the network structure, which will determine both the variance portion (through $W$, $L$, and $U$) and the approximation error. With these set, the bound becomes operational upon choosing $\gamma$, which can be optimized as desired, and this will immediately then yield a convergence rate. 

Turning to special cases, we first show that the optimal rate of \cite{stone1982optimal} can be attained, up to log factors. However, this relies on a rather artificial network structure, designated to approximate functions in a Sobolev space well, but without concern for practical implementation. Thus, while the following rate improves upon Theorem \ref{thm:mlp rate}, we view this result as mainly of theoretical interest: establishing that (certain) deep ReLU networks are able to attain the optimal rate. 

\begin{corollary}[Optimal Rate]
	\label{thm:optimal rate}
	Suppose Assumptions \ref{asmpt:dgp} and \ref{asmpt:sobolev} hold. Let $\widehat{f}_{\rm OPT}$ solve \eqref{eqn:erm} using the (deep and wide) network of \citet[][Theorem 1]{Yarotsky2017_NN}, with $W \asymp U \asymp n^{\frac{d}{2\beta + d}} \log n$ and depth $L \asymp \log n$, the following hold with probability at least $1 - e^{-\gamma}$, for $n$ large enough,
	\begin{enumerate}
		\item $\| \displaystyle \widehat{f}_{\rm OPT} - f_* \|_{L_2(\bx)}^2  \leq  C \cdot \left\{ n^{-\frac{2\beta}{2\beta+d}} \log^4 n + \frac{\log \log n + \gamma}{n} \right\}$ \ and 
		\item $\displaystyle \E_n  \left[ ( \widehat{f}_{\rm OPT} - f_* )^2 \right]  \leq C \cdot \left\{ n^{-\frac{2\beta}{2\beta+d}} \log^4 n + \frac{\log \log n + \gamma}{n} \right\} $,
	\end{enumerate}
	for a constant $C>0$ independent of $n$, which may depend on $d$, $M$, and other fixed constants.
\end{corollary}

Next, we turn to \emph{very} deep networks that are very narrow, which have attracted substantial recent interest. Theorem \ref{thm:mlp rate} and Corollary \ref{thm:optimal rate} dealt with networks where the depth and the width grow with sample size. This matches the most common empirical practice, and is what we use in Sections \ref{sec:application} and \ref{sec:simuls}. However, it is possible to allow for networks of \emph{fixed} width, provided the depth is sufficiently large. The next result is perhaps the largest departure from the classical study of neural networks: earlier work considered networks with diverging width but fixed depth (often a single layer), while the reverse is true here. The activation function is of course qualitatively different as well, being piecewise linear instead of smooth. Using recent results \citep{mhaskar2016deep,Hanin2017_WP,Yarotsky2018_WP} we can establish the following rate for very deep, fixed-width MLPs.
		
\begin{corollary}[Fixed Width Networks]	
	\label{thm:fixed width}
	Let the conditions of Theorem \ref{thm:mlp rate} hold, with $\beta \geq 1$ in Assumption \ref{asmpt:sobolev}. Let $\widehat{f}_{\rm FW}$ solve \eqref{eqn:erm} for an MLP with fixed width $H = 2d+10$ and depth $L \asymp n^{\frac{d}{2(2+d)}}$. Then with probability at least $1 - e^{-\gamma}$, for $n$ large enough,
	\begin{enumerate}
		\item $\displaystyle  \| \widehat{f}_{\rm FW} - f_* \|_{L_2(\bx)}^2  \leq  C \cdot \left\{ n^{-\frac{2}{2+d}} \log^2 n + \frac{\log \log n + \gamma}{n} \right\}$ \ and 
		\item $\displaystyle \E_n \left[ ( \widehat{f}_{\rm FW} - f_* )^2 \right] \leq C \cdot \left\{ n^{-\frac{2}{2+d}} \log^2 n + \frac{\log \log n + \gamma}{n} \right\} $,
	\end{enumerate}
	for a constant $C>0$ independent of $n$, which may depend on $d$, $M$, and other fixed constants.
\end{corollary}

This result is again mainly of theoretical interest. The class is only able to approximate well functions with $\beta = 1$ (cf. the choice of $L$) which limits the potential applications of the result because, in practice, $d$ will be large enough to render this rate, unlike those above, too slow for use in later inference procedures. In particular, if $d \geq 3$, the sufficient conditions of Theorem \ref{thm:ate} fail.

Finally, as mentioned following Theorem \ref{thm:mlp rate}, our theory here will immediately yield a faster rate upon discovery of improved approximation power of this class of networks. In other words, for example, if a proof became available that fixed-width, very deep networks can approximate $\beta$-smooth functions (as in Assumption \ref{asmpt:sobolev}), then Corollary \ref{thm:fixed width} will trivially be improvable to match the rate of Theorem \ref{thm:mlp rate}. Similarly, if the MLP architecture can be shown to share the approximation power with that of Corollary \ref{thm:optimal rate}, then Theorem \ref{thm:mlp rate} will itself deliver the optimal rate. Our proofs will not require adjustment.

\begin{remark}
	\label{rem:computation}
	Although there has been a great deal of work in easing implementation (optimization and tuning) of deep nets, it still may be a challenge in some settings, particularly when using non-standard architectures. See also Remark \ref{rem:regularize}. Given the renewed interest in deep networks, this is an area of study already \citep{hartford2017deep,polson2018posterior} and we expect this to continue and that implementations will rapidly evolve. This is perhaps another reason that Theorem \ref{thm:mlp rate} is, at the present time, the most practically useful, but that (as just discussed) Theorem \ref{thm:generic rates} will be increasingly useful in the future.
\end{remark}

\begin{remark}
	\label{rem:activation}
	
	Our results can be extended easily to include piecewise linear activation functions beyond ReLU. Intuitively, being itself piecewise linear, appropriate combinations of a fixed number of ReLU functions can equal a piecewise linear function (with a fixed number of knots) and therefore the complexity and approximation power can be easily adjusted to this case. See \cite{Bartlett-etal2017_COLT}. 
	
	In principle, similar rates of convergence could be attained for other activation functions, given results on their approximation error. However, it is not clear what practical value would be offered due to computational issues (in which the activation choice plays a crucial role). Indeed, the recent switch to ReLU stems not from their greater approximation power, but from the fact that optimizing a deep net with sigmoid-type activation is unstable or impossible in practice. Thus, while it is certainly possible that we could complement the single-layer results with rates for sigmoid-based deep networks, these results would have no consequences for real-world practice.
	
	From a purely practical point of view, several variations of the ReLU activation function have been proposed recently (including the so-called Leaky ReLU, Randomized ReLU, (Scaled) Exponential Linear Units, and so forth) and have been found in some experiments to improve optimization properties. It is not clear what theoretical properties these activation functions have or if the computational benefits persist more generically, though this area is rapidly evolving. We conjecture that our results could be extended to include these activation functions.
\end{remark}

\section{Parameters of Interest}
	\label{sec:params}

We will use the results above, in particular Theorem \ref{thm:mlp rate}, coupled with results in the semiparametric literature, to deliver valid asymptotic inference for causal effects. The novelty of our results is not in this semiparametric stage per se, but rather in delivering valid inference after relying on deep learning for the first step estimation. In this section we define the parameters of interest, while asymptotic inference is discussed next. 

We will focus, for concreteness, on causal parameters that are of interest across different disciplines: average treatment effects, expected utility (or profits) under different targeting policies, average effects on (non-)treated subpopulations, and decomposition effects. Our focus on causal inference with observational data is due to the popularity of these estimands both in applications and in theoretical work, thus allowing our results to be put to immediate use and easily compared to prior literature. The average treatment effect in particular is often used as a benchmark parameter for studying inference following machine learning (see references in the Introduction). However, armed with our results for deep neural networks we can cover a great deal more (some discussion is in Section \ref{sec:other params}).

The estimation of average causal effects is a well-studied problem, and we will give only a brief overview here. Recent reviews and further references are given by \citet{Belloni-etal2017_Ecma,Athey-Imbens-Pham-Wager2017_AERPP,Abadie-Cattaneo2018_ARE}. We consider the standard setup for program evaluation with observational data: we observe a sample of $n$ units, each exposed to a binary treatment, and for each unit we observe a vector of pre-treatment covariates, $\bX \in \R^d$, treatment status $T \in \{0,1\}$, and a scalar post-treatment outcome $Y$. The observed outcome obeys $Y = T Y(1) + (1-T)Y(0)$, where $Y(t)$ is the (potential) outcome under treatment status $t \in \{0,1\}$. The ``fundamental problem'' is that only $Y(0)$ or $Y(1)$ is observed for each unit, never both.

The crucial identification assumptions, which pertain to all the parameters we consider, are selection on observables, also known as ignorability, unconfoundedness, missingness at random, or conditional independence, and overlap, or common support. Let $p(\bx) = \P[T = 1 \vert \bX = \bx]$ denote the propensity score and $\mu_t(\bx) = E[Y(t) \vert \bX=\bx], \ t \in \{0,1\}$ denote the two outcome regression functions. We then assume the following throughout, beyond which, we will mostly need only regularity conditions for inference.
\begin{assumption}
	\label{asmpt:ignorability} \
	For $t \in \{0,1\}$ and almost surely $\bX$, $\E[Y(t) \vert T, \bX=\bx] = \E[Y(t) \vert \bX=\bx]$ and $\bar{p} \leq p(\bx) \leq 1 - \bar{p}$ for some $\bar{p} > 0$.
\end{assumption}

It will be useful to divide our discussion between parameters that are fully marginal averages, such as the average treatment effect, and those which are for specific subpopulations. Here, ``subpopulations'' refer to the treated or nontreated groups, with corresponding parameters such as the treatment effect for the treated. Any parameter, in either case, can be studied for a suitable subpopulation defined by the covariates $\bX$, such as a specific demographic group. Though causal effects as a whole share some structure, there are slight conceptual and notational differences. In particular, the form of the efficient influence function and doubly robust estimator is different for the two sets, but common within.

\subsection{Full-Population Average Effect Parameters}
	\label{sec:ate}

Here we are interested in averages over the entire population. The prototypical parameter of interest is the average treatment effect:
\begin{equation}
	\label{eqn:ate}
	\tau = \E[Y(1) - Y(0)].
\end{equation}
In the context of our empirical example, the treatment is being mailed a catalog and the outcome is dollars spent (results for the binary purchase decision are available on request).  The average treatment effect, also referred to as ``lift'' in digital contexts, corresponds to the expected gain in revenue from an average individual receiving the catalog compared to the same person not receiving the catalog.

A closely related parameter of interest is the average realized outcome, which in general may be interpreted as the expected utility or welfare from a \emph{counterfactual} treatment policy. In the context of our empirical application this is expected profits; in a medical context it would be the total health outcome. The question of interest here is whether a change in the treatment policy would be beneficial in terms of increasing outcomes, and this is judged using observational data. Intuitively, the average treatment effect is the expected gain from treating the ``next'' person, relative to if they had not been exposed. That is, it is the expected change in the outcome. Expected utility/profit, on the other hand, is concerned with the total outcome, not the difference in outcomes. In the context of our empirical application, we are interested in total sales rather than the change in sales. Our discussion is grounded in this language for easy comparison.

The parameter depends on a counterfactual/hypothetical treatment targeting strategy, which is often itself the object of evaluation. This is simply a rule that assigns a given set of characteristics (e.g.\ a consumer profile), determined by the covariates $X$, to treatment status: that is, a known function (which may include randomization but is not estimated from the sample) $s(\bx): \supp\{\bX\} \mapsto \{0,1\}$. Note well that this is \emph{not} necessarily the observed treatment: $s(\bx_i) \neq t_i$. The policy maker may wish to evaluate the gain from targeting only a certain subset of customers, a price discrimination strategy, or comparisons of different such policies. Our assumptions, while standard, deliver identification of such counterfactuals at no cost.

The parameter of interest is expected utility, or profit, from a fixed policy, given by
\begin{equation}
	\label{eqn:profit}
	\pi(s) = \E \big[s(\bX)Y(1) + \left(1 - s(\bX)\right) Y(0) \big],
\end{equation}
where we make explicit the dependence on the policy $s(\cdot)$. Compare to Equation \eqref{eqn:ate} and recall that the \emph{observed} outcome obeys $Y = T Y(1) + (1-T)Y(0)$. Whereas $\tau$ is the gain in assigning the next person to treatment and is given by the difference in potential outcomes, $\pi(s)$ is the expected outcome that would be observed for the next person if the treatment rule were $s(\bx)$.

A natural question is whether a candidate targeting strategy, say $s'(\bx)$, is superior to baseline or status quo policy, $s_0(\bx)$. This amounts to testing the hypothesis $H_0: \pi(s') \geq \pi(s_0)$. To evaluate this, we can study the difference in expected profits, which amounts to
\begin{equation}
	\label{eqn:profit diff}
	\pi(s', s_0) = \pi(s') - \pi(s_0) = \E \big[(s'(\bX) - s_0(\bX)) Y(1) + \left(s_0(\bX) - s'(\bX)\right) Y(0) \big] .
\end{equation}
Assumption \ref{asmpt:ignorability} provides identification for $\pi(s)$ and $\pi(s', s_0)$, arguing analogously as for $\tau$. Moreover, notice that $\pi(s', s_0)  = \E [(s'(\bX) - s_0(\bX)) (Y(1) -  Y(0) )] = \E [(s'(\bX) - s_0(\bX)) \tau(\bX)]$, where $\tau(\bx) = \E[Y(1) -  Y(0) \mid \bX = \bx]$ is the conditional average treatment effect. The latter form makes clear that only those differently treated, of course, impact the evaluation of $s'$ compared to $s_0$. The strategy $s'$ will be superior if, on average, it targets those with a higher individual treatment effect. Estimating the optimal treatment policy from the data is discussed briefly in Section \ref{sec:policy}.

The common structure of these parameters is that they all involve full-population averages of the potential outcomes, possibly scaled by a known function. For these parameters, the influence function is known from \cite{Hahn1998_Ecma}, and estimators based on the influence function are doubly robust, as they remain consistent if either the regression functions or the propensity score are correctly specified \citep{Robins-Rotnitzky-Zhao1994_JASA,Robins-Rotnitzky-Zhao1995_JASA}. With a slight abuse of terminology (since we are omitting the centering), the influence function for a single average potential outcome, $t \in \{0,1\}$, is given by, for $\bz = (y, t, \bx')'$,
\begin{equation}
	\label{eqn:ate influence}
	\psi_t(\bz) = \frac{  \One\{T = t\} ( y - \mu_t(\bx)) }{\P[T = t \mid \bX = \bx]}  +  \mu_t(\bx).
\end{equation}  
Our estimation of $\tau$, $\pi(s)$, and $\pi(s', s_0)$ will utilize sample averages of this function, with unknown objects replaced by estimators. Our use of influence functions here follows the recent literature in econometrics showing that the double robustness implies valid inference under weaker conditions on the first step nonparametric estimates \citep{Farrell2015_arXiv,Chernozhukov-etal2018_EJ}.

\subsection{Subpopulation Effect Parameters}
	\label{sec:tot}

The second type of causal effects of interest are based on potential outcomes averaged over only a specific treatment group. A single such average, for $t, t' \in \{0,1\}$, is denoted by
\begin{equation}
	\label{eqn:rho}
	\rho_{t,t'} = \E[Y(t) \mid T = t'].
\end{equation}
Many interesting parameters are linear combinations of these for different $t$ and $t'$. We focus on two for concreteness. (We could also consider averages restricted by targeting-type functions, as in expected utility/profit, but for brevity we omit this.) The most well-studied of these parameters is the treatment effect on the treated, given by
\begin{equation}
	\label{eqn:tot}
	\tau_{1,0} = \E[Y(1) - Y(0) \mid T = 1] = \rho_{1,1} - \rho_{0,1}.
\end{equation}

To appreciate the breadth of this framework, and the applicability of our causal inference results, we also consider a decomposition parameter, a semiparametric analogue of Oaxaca-Blinder \citep{Kitagawa1955_JASA,Oaxaca1973_IER,Blinder1973_JHR}. In this context, the ``treatment'' variable $T$ is typically not a treatment assignment per se, but rather an exogenous covariate such as a demographic indicator, perhaps most commonly a male/female indicator.  See \cite{Fortin-Lemieux-Firpo2011_handbook} for a complete discussion and further references. The parameter of interest in this case is the decomposition of $\Delta = \E[Y(1) \mid T = 1]  - \E[Y(0) \mid T = 0]$, into the difference in the covariate distributions and the difference in expected outcomes. These can be written as functions of different $\rho_{t,t'}$. For example, $\Delta_X = \E[Y(1) \vert T \!=\! 1]  - \E[Y(1) \vert T \!=\! 0]  = \E[\mu_1(\bX) \vert T \!=\! 1]  - \E[\mu_1(\bX) \vert T \!=\! 0]  = \rho_{1,1} - \rho_{1,0}$. We are in general interested in
\begin{equation}
	\label{eqn:oaxaca}
	\Delta = \Delta_X + \Delta_\mu,   		\qquad \quad  		   \Delta_X = \rho_{1,1} - \rho_{1,0},   		\qquad \text{and} \qquad  		  \Delta_\mu = \rho_{1,0} - \rho_{0,0}.
\end{equation}

Just as in the case of full-population averages, the influence function is known and leads to a doubly robust estimator. For a single $\rho_{t,t'}$, the (uncentered) influence function is (cf.\ \eqref{eqn:ate influence}):
\begin{equation}
	\label{eqn:tot influence}
	\psi_{t,t'}(\bz) =   \frac{\P[T = t' \mid \bX = \bx]}{\P[T = t']} \frac{  \One\{T = t\} ( y - \mu_t(\bx)) }{\P[T = t \mid \bX = \bx]}  +  \frac{\One\{T = t'\} \mu_t(\bx)}{\P[T = t']}.
\end{equation}
Estimation and inference requires, as above, estimation of the propensity scores and regression functions, depending on the exact choices of $t$ and $t'$, and here we also require the marginal probability of treatments.

\subsection{Optimal Policies}
	\label{sec:policy}

Moving beyond a fixed parameter, our results on deep neural networks can be used to address optimal targeting. In the notation of Section \ref{sec:ate}, this amounts to finding a policy, say $s_\star(\bx)$, that maximizes a given measure of utility stemming from treatment, generally the expected gain relative to a baseline policy. In Section \ref{sec:ate} we considered the utility (or profit) difference between two given strategies, a candidate $s'(\bx)$ and a baseline $s_0(\bx)$. Instead of inference on $\pi(s', s_0)$, we can use the data to find the $s_\star(\bx)$ which maximizes the gain relative to the baseline. This problem has been widely studied in econometrics and statistics; for detailed discussion and numerous references see \cite{Manski2004_Ecma}, \cite{Hirano-Porter2009_Ecma}, \cite{Kitagawa-Tetenov2018_Ecma}, and \cite{Athey-Wager2018_WP}. In particular, the latter noticed using the locally robust framework allows policy optimization under nearly the same conditions as inference and proved fast convergence rates of the estimated policy in terms of regret.

More formally, we want to find the optimal choice $s_\star(\bx)$ in some policy/action space $\mathcal{S}$. The policy space, and thus its complexity, is user determined. Simple examples include simple decision trees or univariate-based strategies; more can be found in the references above. Recall that $\pi(s', s_0) = \E[Y(s')] - \E[Y(s_0)] = \E [(s'(\bX) - s_0(\bX)) \tau(\bX)]$, where $\tau(\bx) = \E[Y(1) -  Y(0) \mid \bX = \bx]$ is the conditional average treatment effect. Given a space $\mathcal{S}$, we wish to find the policy $s_\star(\bx) \in \mathcal{S}$ which solves $\max_{s' \in \mathcal{S}} \pi(s', s_0)$. The main result of \cite{Athey-Wager2018_WP} is that replacing $\pi$ with the doubly-robust $\hat{\pi}$ of Equation \eqref{eqn:dr}, and minimizing the empirical analogue of regret, one obtains an estimator $\hat{s}(\bx)$ of the optimal policy that obeys the regret bound $\pi(s_\star, s_0) - \pi(\hat{s}, s_0) = O_P(\sqrt{{\rm VC}(\mathcal{S})/n})$ (a formal statement would be notationally burdensome). The complexity of the user-chosen policy space enters the bound through its VC dimension. Simple, interpretable policy classes often have bounded or slowly-growing dimension, implying rapid convergence.

\subsection{Other Estimands}
	\label{sec:other params}

There are of course many other contexts where first-step deep learning is useful. Only trivial extensions to the above would be required for other causal effects, such as multi-valued treatments \citep[reviewed by][]{Cattaneo2010_JoE} and others with doubly-robust estimators \citep{Sloczynski-Wooldridge2018_ET}. Further, under selection on observables, treatment effects, missing data, measurement error, and data combination are equivalent, and thus all our results apply immediately to those contexts. For reviews of these and Assumption \ref{asmpt:ignorability} more broadly, see \cite{Chen-Hong-Tarozzi2004_WP,Tsiatis2006_book,Heckman-Vytlacil2007a_Handbook,Imbens-Wooldridge2009_JEL}.

Moving beyond causal effects, any estimand with a locally/doubly robust estimator depending only on target functions falling into our class of losses can be covered using the results of Section \ref{sec:deep nets}. For example, estimands requiring distribution estimation require further study; see \cite{Liang2018_GAN} for recent results via Generative Adversarial Networks (GANs). More precisely, along with regularity conditions, our theory can be used to verify the conditions of \cite{Chernozhukov-etal2018_WP}, who treat more general semiparametric estimands using local robustness, sometimes relying on sample splitting or cross fitting. Further in this vein, our results on deep neural networks can be used to address optimal targeting, i.e., finding the policy, say $s_\star(\bx)$, that maximizes a given measure of utility, by applying the results of \cite{Athey-Wager2018_WP}, who noticed that using the locally robust framework allows policy optimization.

More broadly, the learning of features using deep neural networks is becoming increasingly popular and our results speak to this context directly. To illustrate, consider the simple example of a linear model where some predictors are features learned from independent data. Here, the object of interest is the fixed-dimension coefficient vector $\bm{\lambda}$, which we assume can be partitioned as $\bm{\lambda} = (\bm{\lambda}_1', \bm{\lambda}_2')'$ according to the model $Y = \bm{f}(\bX)'\bm{\lambda}_1 + W'\bm{\lambda}_2 + \varepsilon$. The features $\bm{f}(\bX)$, often a ``score'' of some type, are generally learned from auxiliary (and independent) data. For a recent example, see \cite{liu2017large}. In such cases, inference on $\bm{\lambda}$ can proceed directly, as long as care is taken to interpret the results. See Section \ref{sec:splitting}.

\section{Asymptotic Inference}
	\label{sec:inference}

We now turn to asymptotic inference for the causal parameters discussed above. We first define the estimators, which are based on sample averages of the (uncentered) influence functions \eqref{eqn:ate influence} and \eqref{eqn:tot influence}. We then give a generic result for single averages which can then be combined for inference on a given parameter of interest. Below we discuss inference under randomized treatment and using sample splitting. 

Throughout, we assume we have a sample $\{\bz_i = (y_i, t_i, \bx_i')'\}_{i=1}^n$ from $\bZ = (Y, T, \bX')'$. We then form
\begin{equation}
	\label{eqn:influence hat}
	\hat{\psi}_t(\bz_i) = \frac{ \One\{t_i = t\} ( y_i - \hat{\mu}_t(\bx_i)) }{\hat{\P}[T = t \mid \bX = \bx_i]}  +  \hat{\mu}_t(\bx_i),
\end{equation}
where $\hat{\P}[T = t \mid \bX = \bx_i] = \hat{p}(\bx_i)$ for $t=1$ and $1-\hat{p}(\bx_i)$ for $t=0$, and similarly
\begin{equation}
	\label{eqn:influence hat tot}
	\hat{\psi}_{t,t'}(\bz_i) =   \frac{\hat{\P}[T = t' \mid \bX = \bx_i]}{\hat{\P}[T = t']} \frac{  \One\{t_i = t\} ( y_i - \hat{\mu}_t(\bx_i)) }{\hat{\P}[T = t \mid \bX = \bx_i]}  +  \frac{\One\{t_i = t'\} \hat{\mu}_t(\bx_i)}{\hat{\P}[T = t']},
\end{equation}
where $\hat{\P}[T = t']$ is simply the sample frequency $\En[\One\{t_i = t'\}]$.

For the first stage estimates appearing in \eqref{eqn:influence hat} and \eqref{eqn:influence hat tot} we use our results on deep nets, and Theorem \ref{thm:mlp rate} in particular. Specifically, the estimated propensity score, $\hat{p}(\bx)$, is the estimate that results from solving \eqref{eqn:erm}, with the MLP architecture, for the logistic loss \eqref{eqn:logit} with $T$ as the outcome. Similarly, for each status $t \in \{0,1\}$, we can let $\hat{\mu}_t(\bx)$ be the deep-MLP estimate of $f_*(\bx) = \E[ Y | T = t, \bX = \bx]$, solving \eqref{eqn:erm} for least squares loss, \eqref{eqn:ols}, with outcome $Y$, using only observations with $t_i = t$. However, it is worth noting that the theoretically-equivalent joint estimation of Equation \eqref{eqn:joint} performs much better, as the two groups may share features. To state the results, let $\beta_p$ and $\beta_\mu$ be the smoothness parameters of Assumption \ref{asmpt:sobolev} for the propensity score and outcome models, respectively.

We then obtain inference using the following results, essentially taken from \cite{Farrell2015_arXiv}. Similar results are given by \cite{Belloni-etal2017_Ecma} and \cite{Chernozhukov-etal2018_EJ}. All of these provide high-level conditions for valid inference, and none verify these for deep nets as we do here.
\begin{theorem}
	\label{thm:ate}
	Suppose that $\{\bz_i = (y_i, t_i, \bx_i')'\}_{i=1}^n$ are i.i.d.\ obeying Assumption \ref{asmpt:ignorability} and the conditions Theorem \ref{thm:mlp rate} hold with $\beta_p \wedge \beta_\mu > d$. Further assume that, for $t \in \{0,1\}$, $\E[(s(\bX)\psi_t(\bZ))^2 \vert \bX]$ is bounded away from zero and, for some $\delta > 0$, $\E[(s(\bX)\psi_t(\bZ))^{4+\delta} \vert \bX]$ is bounded. Then the deep MLP-ReLU network estimators defined above obey the following, for $t \in \{0,1\}$,
		\begin{enumerate}
	
			\item $ \E_n[(\hat{p}(\bx_i) - p(\bx_i))^2]  = o_P(1)$ and $\E_n \left[ (\hat{\mu}_t(\bx_i) - \mu_t(\bx_i))^2\right]  = o_P(1)$,
	
			\item $ \E_n[ (\hat{\mu}_t(\bx_i) - \mu_t(\bx_i))^2]^{1/2}  \E_n[ (\hat{p}(\bx_i) - p(\bx_i))^2]^{1/2}    = o_P(n^{-1/2})$, and 
			
			\item $ \E_n[ (\hat{\mu}_t(\bx_i) - \mu_t(\bx_i)) (1 - \One\{t_i=t\}/ \P[T=t\vert \bX = \bx_i])]    = o_P(n^{-1/2})$,
	
		\end{enumerate}
	and therefore, if $\hat{p}(\bx_i)$ is bounded inside $(0,1)$, for a given $s(\bx)$ and $t \in \{0,1\}$, we have
	\[
		\sqrt{n} \E_n \left[ s(\bx_i)\hat{\psi}_t(\bz_i) - s(\bx_i)\psi_t(\bz_i) \right] = o_P(1)   		  \quad  		  \text{ and }  		  \quad  		\frac{\E_n[ (s(\bx_i)\hat{\psi}_t(\bz_i))^2]}{\E_n[ (s(\bx_i)\psi_t(\bz_i))^2]} = o_P(1),  
	\]
	as well as,
 	\[
 		\sqrt{n} \E_n \left[ \hat{\psi}_{t,t'}(\bz_i) - \psi_{t,t'}(\bz_i) \right] = o_P(1)   		  \quad  		  \text{ and }  		  \quad  		\frac{\E_n[ \hat{\psi}_{t,t'}(\bz_i)^2]}{\E_n[ \psi_{t,t'}(\bz_i)^2]} = o_P(1).  
 	\]	
\end{theorem}

This result, our main inference contribution, shows exactly how deep learning delivers valid asymptotic inference for our parameters of interest. Theorem \ref{thm:mlp rate} (a generic result using Theorem \ref{thm:generic rates} could be stated) proves that the nonparametric estimates converge sufficiently fast, as formalized by conditions (a), (b), and (c), enabling feasible efficient semiparametric inference. In general, these are implied by, but may be weaker than, the requirement of that the first step estimates converge faster than $n^{-1/4}$, which our results yield for deep ReLU nets. The first is a mild consistency requirement. The second requires a rate, but on the product of the two estimates, which can be satisfied under weaker conditions. Finally, the third condition is the strongest. Intuitively, this condition arises from a ``leave-in'' type remainder, and as such, it can be weakened using sample splitting \cite{Chernozhukov-etal2018_EJ,Newey-Robins2018_WP}. We opt to maintain (c) exactly because deep nets are not amenable to either simple leave-one-out forms (as are, e.g., classical kernel regression) or to sample splitting, being a data hungry method the gain in theoretically weaker rate requirements may not be worth the price paid in constants in finite samples. Instead, we employ our localization analysis, as was used to obtain the results of Section \ref{sec:deep nets}, to verify (c) directly (see Lemma \ref{lem:logit}); this appears to be a novel application of localization, and this approach may be useful in future applications of second-step inference using machine learning methods.

From this result we immediately obtain inference for all the causal parameters discussed above. For the full-population averages, for example, we would form
\begin{align}
	\begin{split}
		\label{eqn:dr}
		\hat{\tau} & = \E_n \left[ \hat{\psi}_1(\bz_i) - \hat{\psi}_0(\bz_i) \right],   			\\
		\hat{\pi}(s) & = \E_n \left[ s(\bx_i)\hat{\psi}_1(\bz_i)  + (1-s(\bx_i))\hat{\psi}_0(\bz_i) \right],   			\\
		\hat{\pi}(s', s_0) & = \E_n \left[ [s'(\bx_i) - s_0(\bx_i)] \hat{\psi}_1(\bz_i) - [s'(\bx_i) - s_0(\bx_i)] \hat{\psi}_0(\bz_i) \right].
	\end{split}
\end{align}
The estimator $\hat{\tau}$ is exactly the doubly/locally robust estimator of the average treatment effect that is standard in the literature. The estimators for profits can be thought of as the doubly robust version of the constructs described in \cite{Hitsch-Misra2018_WP}. Furthermore, to add a per-unit cost of treatment/targeting $c$ and a margin $m$, simply replace $\psi_1$ with $m \psi_1 - c$ and $\psi_0$ with $m \psi_0$. Similarly, $\hat{\tau}_{1,0}$, $\hat{\Delta}_X$, and $\hat{\Delta}_\mu$ would be linear combinations of different $\hat{\rho}_{t,t'} = \E_n [ \hat{\psi}_{t,t'}(\bz_i) ]$.

It is immediate from Theorem \ref{thm:ate} that all such estimators are asymptotically Normal. The asymptotic variance can be estimated by simply replacing the sample first moments of \eqref{eqn:dr} with second moments. That is, looking at $\hat{\pi}(s)$ to fix ideas,
\[ \sqrt{n} \hat{\Sigma}^{-1/2}\left(\hat{\pi}(s) - \pi(s) \right)  \stackrel{d}{\to} \mathcal{N}(0,1)  ,   \quad \text{with}\quad   \hat{\Sigma} =   \E_n \left[ \left(  s(\bx_i)\hat{\psi}_1(\bz_i)  + (1-s(\bx_i))\hat{\psi}_0(\bz_i) \right)^2 \right] -  \hat{\pi}(s)^2.  \]
The others are similar. Further, Theorem \ref{thm:ate} can be generalized straightforwardly to yield uniformly valid inference, following the approach of \cite{Romano2004_SJS}, exactly as in \cite{Belloni-Chernozhukov-Hansen2014_REStud} or \cite{Farrell2015_arXiv}.

Finally, we note that our focus with Theorem \ref{thm:ate} is showcasing the practical utility of deep learning. Our use of local/double robustness here is toward the aim of attaining feasible inference without requiring more detailed assumptions on the machine learning step. This comes at the expense of, for example, stronger-than-minimal smoothness assumptions. That is, the requirement that $\beta_p \wedge \beta_\mu > d$ is not minimal, and moreover, neither is the weaker condition $\beta_p \wedge \beta_\mu > d/2$ that would be required after applying Corollary \ref{thm:optimal rate} instead of Theorem \ref{thm:mlp rate}. Obtaining a Gaussian limit, and possibly semiparametric efficiency, under minimal conditions has been studied by many, dating at least to \cite{Bickel-Ritov1988_Sankhya}; see \cite{Robins-etal2009_EJS} for recent results and references on optimal estimation and minimal conditions. For causal inference, \cite{Chen-Hong-Tarozzi2008_AoS} and \cite{Athey-Imbens-Wager2018_JRSSB} obtain semiparametric efficiency under strictly weaker conditions than ours on $p(\bx)$ (the former under minimal smoothness on $\mu_t(\bx)$ and the latter under a sparsity in a high-dimensional linear model). Further, as above, cross-fitting \citep{Newey-Robins2018_WP} paired with local robustness may yield weaker smoothness conditions by providing underfitting'' robustness (i.e.\ weakening bias-related tuning parameter assumptions). On the other hand, weaker variance-related assumptions, or ``overfitting'' robust inference procedures, \citep{Cattaneo-Jansson2018_Ecma,Cattaneo-Jansson-Ma2018_RESTUD}, may also be possible following deep learning, but are less automatic at present. Finally, other methods designed for causal inference under relaxed assumptions may be useful here, such as the recently developed extensions to doubly robust estimation \citep{tan2018model} and inverse weighting \citep{Ma-Wang2018_IPW}: pursuing these in the context of deep learning is left to future work.

\subsection{Inference Under Randomization}
	\label{sec:rct}

Our analysis thus far has focused on observational data, but it is worth spelling out results for randomized experiments. This is particularly important in the Internet age, where experimentation is common, vast amounts of data are available, and effects are often small in magnitude \citep{Taddy-etal2015_WP}. Indeed, our empirical illustration, detailed in the next section, stems from an experiment with 300,000 units and hundreds of covariates.  When treatment is randomized, inference can be done directly using the mean outcomes in the treatment and control groups, such as the difference for the average treatment effect or the corresponding weighted sum for profit. However, pre-treatment covariates can be used to increase efficiency \citep{Hahn2004_REStat}.

We will focus on the simple situation of a purely randomized binary treatment, but our results can be extended naturally to other randomization schemes. We formalize this with the following.

\begin{assumption}[Randomized Treatment]
	\label{asmpt:rct} \
	$T$ is independent of $Y(0)$, $Y(1)$, and $X$, and is distributed Bernoulli with parameter $p^*$, such that $\bar{p} \leq p^* \leq 1 - \bar{p}$ for some $\bar{p} > 0$.
\end{assumption}

Under this assumption, the obvious simplification is that the propensity score need not be estimated using the covariates, but can be replaced with the (still nonparametric) sample frequency: $\hat{p}(\bx_i) \equiv \hat{p} = \En [t_i]$. This is plugged into Equation \eqref{eqn:dr} and estimation and inference proceeds as above. Only rate conditions on the regression functions $\hat{\mu}_t(x)$ are needed. Further, conditions (a) and (b) of Theorem \ref{thm:ate} collapse, as $\hat{p}$ is root-$n$ consistent, leaving only condition (c) to be verified. Again, cross-fitting can be used in theory to remove this condition and thus weaken the requirement that $\beta_\mu > d$, but we maintain this for simplicity. We collect this into the following result, which is a trivial corollary of Theorem \ref{thm:ate}. 
\begin{corollary}
	\label{cor:rct}
	Let the conditions of Theorem \ref{thm:ate} hold with Assumption \ref{asmpt:rct} in place of Assumption \ref{asmpt:ignorability} and only $\beta_\mu > d$. Then deep MLP-ReLU network estimators obey 
		\begin{enumerate}
	
			\item[{\bf (a$^\prime$)}] $\E_n \left[ (\hat{\mu}_t(\bx_i) - \mu_t(\bx_i))^2\right]  = o_P(1)$ and
			
			\item[{\bf (c$^\prime$)}] $ \E_n[ (\hat{\mu}_t(\bx_i) - \mu_t(\bx_i)) (1 - \One\{t_i=t\}/ p^*)]    = o_P(n^{-1/2})$
	
		\end{enumerate}
	and the conclusions of Theorem \ref{thm:ate} hold.
\end{corollary}

\subsection{Sample Splitting}
	\label{sec:splitting}

Sample splitting may be used to obtain valid inference in cases, unlike those above, where the parameter of interest itself is learned from the data. For the causal estimands above, the regression functions and propensity score must be estimated, but these are nuisance functions. This is not true in the inference after policy or feature learning (Sections \ref{sec:policy} and \ref{sec:other params}). For policy learning, our results can be used to verify the high-level conditions of \cite{Athey-Wager2018_WP}, though they require the additional condition of uniform consistency of the first stage estimators, and for machine learning estimators this is not clearly innocuous. However, this gives only point estimation. 

Sample splitting is used in the obvious way: the first subsample, or more generally, independent auxiliary data, is used to learn the features or optimal policy, and then Theorem \ref{thm:ate} is applied in the second subsample, conditional on the results of the first. For policy learning this delivers valid inference on $\pi(\hat{s})$ or $\pi(\hat{s}, s_0)$, while for the simple example of feature learning in a linear model we obtain inference on the parameters defined by the ``model'' $Y = \bm{f}(\bX)'\bm{\lambda}_1 + W'\bm{\lambda}_2 + \varepsilon$, where $\bm{f}(\bX)$ is estimated from auxiliary data. Care must be taken in interpreting the results. The results of the first-subsample estimation are effectively conditioned upon in the inference stage, redefining the target parameter to be in terms of the learned object. In many contexts this may be sufficient \citep{Chernozhukov-etal2018_generic}, but further assumptions will generally be needed to assume that the first subsample has recovered the true population object. To fix ideas, consider policy learning: inference on $\pi(\hat{s}, s_0)$, conditional on the map $\hat{s}(\bx)$ learned in the first subsample, is immediate and requires no additional assumptions, but inference on $\pi(s_\star, s_0)$ is not obvious without further conditions.

\section{Empirical Application}
	\label{sec:application}

To illustrate our results, Theorems \ref{thm:mlp rate} and \ref{thm:ate} in particular, we study, from a marketing point of view, a randomized experiment from a large US retailer of consumer products. The outcome of interest is consumer spending and the treatment is a catalog mailing. The firm sells directly to the customer (as opposed to via retailers) using a variety of channels such as the web and mail. The data consists of nearly three hundred thousand (292,657) consumers chosen at random from the retailer's database. Of these, 2/3 were randomly chosen to receive a catalog, and in addition to treatment status, we observe roughly one hundred fifty covariates, including demographics, past purchase behaviors, interactions with the firm, and other relevant information. For more on the data and a complete discussion of the decision making issues, we refer the reader to \cite{Hitsch-Misra2018_WP} (we use the 2015 sample). That paper studied various estimators, both traditional and modern, of average and heterogeneous causal effects. Importantly, they did not consider neural networks. Our results show that deep nets are at least as good as (and sometimes better than) the best methods in \cite{Hitsch-Misra2018_WP}.

In terms of motivation, a key element of a firm's toolkit is the design and implementation
of targeted marketing instruments. These instruments, aiming to induce demand, often contain advertising and informational content about the firms offerings. 
The targeting aspect thus boils down to the selection of which particular customers should be sent the material. This is a particularly important decision
since the costs of creation and dissemination of the material can accumulate rapidly, particularly over a large customer base. For a typical retailer engaging in direct marketing the costs of sending out a catalog 
can be close to a dollar per targeted customer.
With millions of catalogs being sent out, the cost of a typical campaign is quite high.

Given these expenses, an important problem for firms is ascertaining
the causal effects of such targeted mailing, and then using
these effects to evaluate potential targeting strategies. At a high level, this approach is very similar to modern personalized medicine where
treatments have to be targeted. In these contexts, both the treatment and the targeting can be costly, and thus careful assessment of $\pi(s)$ (interpreted as welfare) is crucial for decision making.

The outcome of interest for the firm is customer spending. This is the
total amount of money that a given customer spends on purchases of
the firm's products, within a specified time window. For the experiment
in question the firm used a window of three months, and aggregated sales from all available purchase channels including phone, mail, and the web. In our data 6.2\% of customers made a purchase. Overall mean spending is \$7.31; average spending conditional on buying is \$117.7, with a standard deviation of \$132.44. The idea then is to examine the incremental effect that the catalog
had on this spending metric. Table \ref{tab:Summary-Statistics} presents summary statistics for the outcome and treatment. Figure \ref{fig:Spend-conditional-on} displays the complete density of spending conditional on a purchase, which is quite skewed. 

\begin{figure}[ht!]
	\centering
	\includegraphics[scale=0.7]{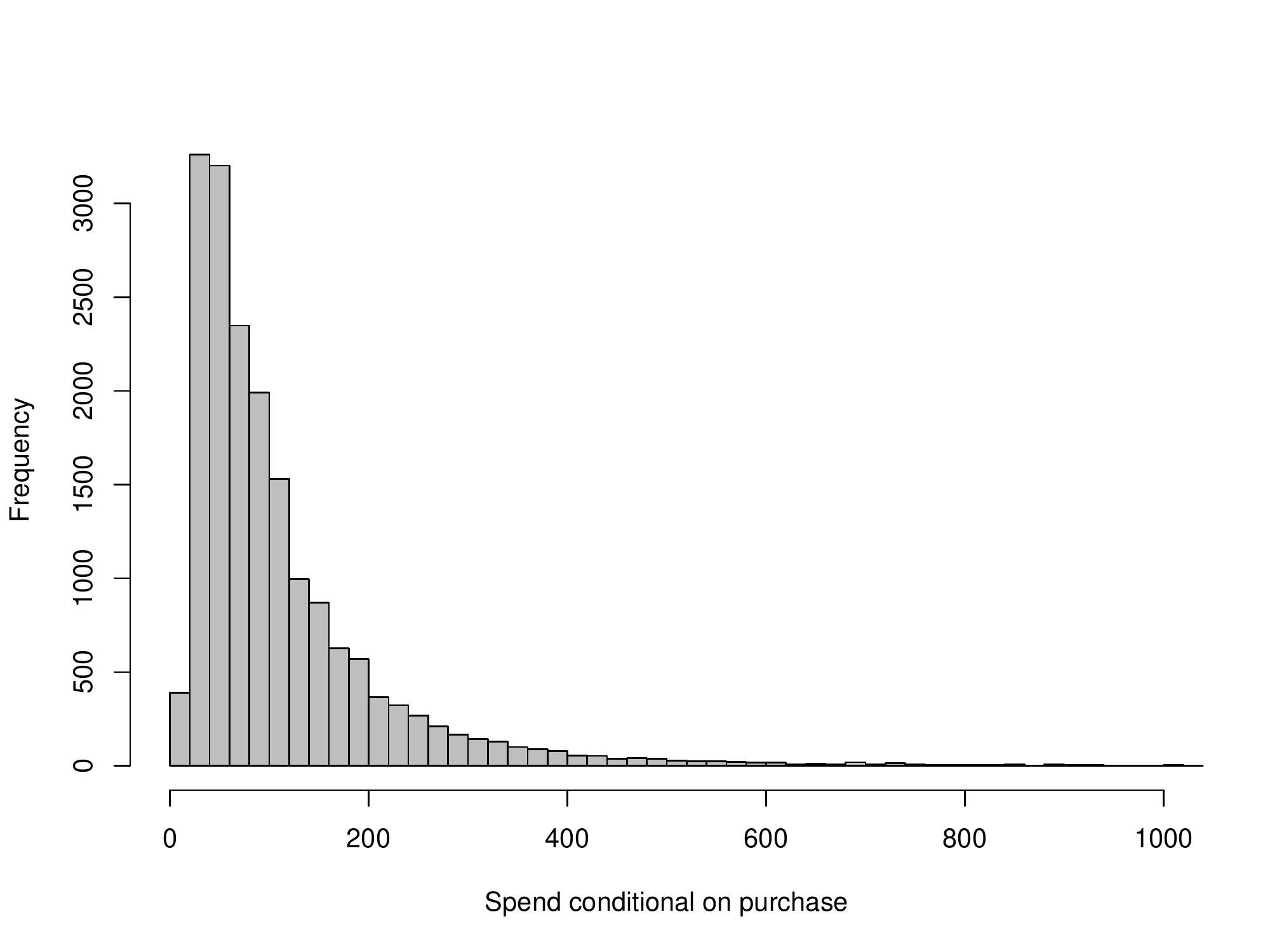}
	\caption{Spend Conditional on Purchase}
	\label{fig:Spend-conditional-on}
\end{figure}

\begin{table}
	\caption{Summary Statistics}
	\label{tab:Summary-Statistics}
	\renewcommand{\arraystretch}{1.2}
	\centering
	\begin{tabular}{rrrr} 
		\hline
		\hline 
		& Mean & SD & N \\
		\hline Purchase & 0.062 & 0.24 & 292657 \\
		Spend & 7.311 & 43.55 & 292657 \\
		Spend Conditional on Purchase & 117.730 & 132.44 & 18174 \\
		Treatment & 0.669 & 0.47 & 292657 \\ 
		Purchase $\mid$ Treatment=1 &  0.069 & 0.25 & 195821 \\
		Purchase $\mid$ Treatment=0 &  0.047 & 0.21 & 96836 \\
		Spend $\mid$ Treatment=1 &  8.158 & 44.71 & 195821 \\
		Spend $\mid$ Treatment=0 &  5.597 & 41.04 & 96836 \\
		\hline
	\end{tabular}
	\renewcommand{\arraystretch}{1}
\end{table}

\subsection{Implementation Details}

We estimated deep neural nets under a variety of architecture choices.
In what follows we present eight examples
and focus on one particular architecture to compute various statistics
and tests to illustrate the use of the theory developed above. All computation was done using TensorFlow\textsuperscript{\tiny TM}.

For treatment effect and profit estimation we follow Equations \eqref{eqn:influence hat} and \eqref{eqn:dr}. Because treatment is randomized, we apply Corollary \ref{cor:rct}, and thus, only require estimates of the regression functions $\mu_t(\bx) = E[Y(t) \vert \bX=\bx], \ t \in \{0,1\}$. An important implementation detail, from a computation point of view (recall Remark \ref{rem:computation}) is that we will estimate $\mu_0(\bx)$ and $\tau(\bx)$ (and thereby $\mu_1(\bx)$) \emph{jointly} (results from separate estimation are available). To be precise, recalling Equations \eqref{eqn:ols} and \eqref{eqn:erm}, we solve
\begin{align}
	\label{eqn:joint}
	\begin{pmatrix}
		\hat{\mu}_0(\bx) \\ 
		\hat{\tau}(\bx) = \hat{\mu}_1(\bx) - \hat{\mu}_0(\bx) 
	\end{pmatrix} 
	\deq \argmin_{\tilde{\mu}_0, \tilde{\tau}}   \sum_{i=1}^n \frac{1}{2} \Big( y_i - \tilde{\mu}_0(\bx_i) - \tilde{\tau}(\bx_i) t_i \Big)^2
\end{align}
where the minimization is over the relevant network architecture. Recall that, in the context of our empirical example $y_i$ is the customer's spending, $\bx_i$ are her characteristics, and $t_i$ indicates receipt of a catalog. In this format, $\mu_0(\bx_i)$ reflects base spending and $\tau(\bx) = \mu_1(\bx_i) - \mu_0(\bx_i)$ is the conditional average treatment effect of the catalog mailing. In our application, this joint estimation outperforms separately estimating each $\mu_t(\bx)$ on the respective samples (though these two approaches are equivalent theoretically).

The details of the eight deep net architectures are presented in Table \ref{tab:Estimated-Deep-Networks}. See Section \ref{sec:construction} for an introduction to the terminology and network construction. Most yielded similar results, both in terms of fit and final estimates. A key measure of fit reported in the final column of the table is the portion of $\hat{\tau}(\bx_i)$ that were negative. As argued by \citet{Hitsch-Misra2018_WP}, it is implausible under standard marketing or economic theory that receipt of a catalog causes lower purchasing. On this metric of fit, deep nets perform as well as, and sometimes better than, the best methods found by \citet{Hitsch-Misra2018_WP}: Causal KNN with Treatment Effect Projections (detailed therein) or Causal Forests \citep{Wager-Athey2018_JASA}. Figure \ref{fig:Conditional-Average-Treatment-1} shows the distribution of $\hat{\tau}(\bx_i)$ across customers for each of the eight architectures. While there are differences in the shapes of the densities, the mean and variance estimates are nonetheless quite similar.

\begin{table}[ht!]
	\renewcommand{\arraystretch}{1.2}
	\caption{\label{tab:Estimated-Deep-Networks} Deep Network Architectures}
	\begin{center}
		\resizebox{\columnwidth}{!}{
			\begin{tabular}{ccccccccc}   
				\hline   
				\hline
				& Learning & Widths &  Dropout &  Total & Validation & Training  & \\  
				Architecture & Rate & [$H_1$, $H_2$, ...] & [$H_1$, $H_2$, ...] & Parameters & Loss & Loss  & $\P_n[\hat{\tau}(\bx_i)<0]$ \\   
				\hline    
				1 & 0.0003 & [60] & [0.5] & 8702 & 1405.62 & 1748.91 & 0.0014 \\
				2 & 0.0003 & [100] & [0.5] & 14502 & 1406.48 & 1751.87 & 0.0251 \\
				3 & 0.0001 & [30, 20] & [0.5, 0] & 4952 & 1408.22 & 1751.20 & 0.0072 \\
				4 & 0.0009 & [30, 10] & [0.3, 0.1] & 4622 & 1408.56 & 1751.62 & 0.0138 \\
				5 & 0.0003 & [30, 30] & [0, 0]  & 5282 & 1403.57 & 1738.59 & 0.0226 \\
				6 & 0.0003 & [30, 30] & [0.5, 0]  & 5282 & 1408.57 & 1755.28 & 0.0066 \\
				7 & 0.0003 & [100, 30, 20] & [0.5, 0.5, 0]  & 17992 & 1408.62 & 1751.52 & 0.0103 \\
				8 & 0.00005 & [80, 30, 20] & [0.5, 0.5, 0]  & 14532 & 1413.70 & 1756.93 & 0.0002 \\
				\hline
			\end{tabular}
		} 
	\end{center}
	\footnotesize\textbf{Notes}: All networks use the ReLU activation function. The width of each layer is shown, e.g.\ Architecture 3 consists of two layers, with 30 and 20 hidden units respectively. The final column shows the portion of estimated individual treatment effects below zero.
\end{table}

\begin{figure}[ht!]
	\includegraphics[scale=0.7]{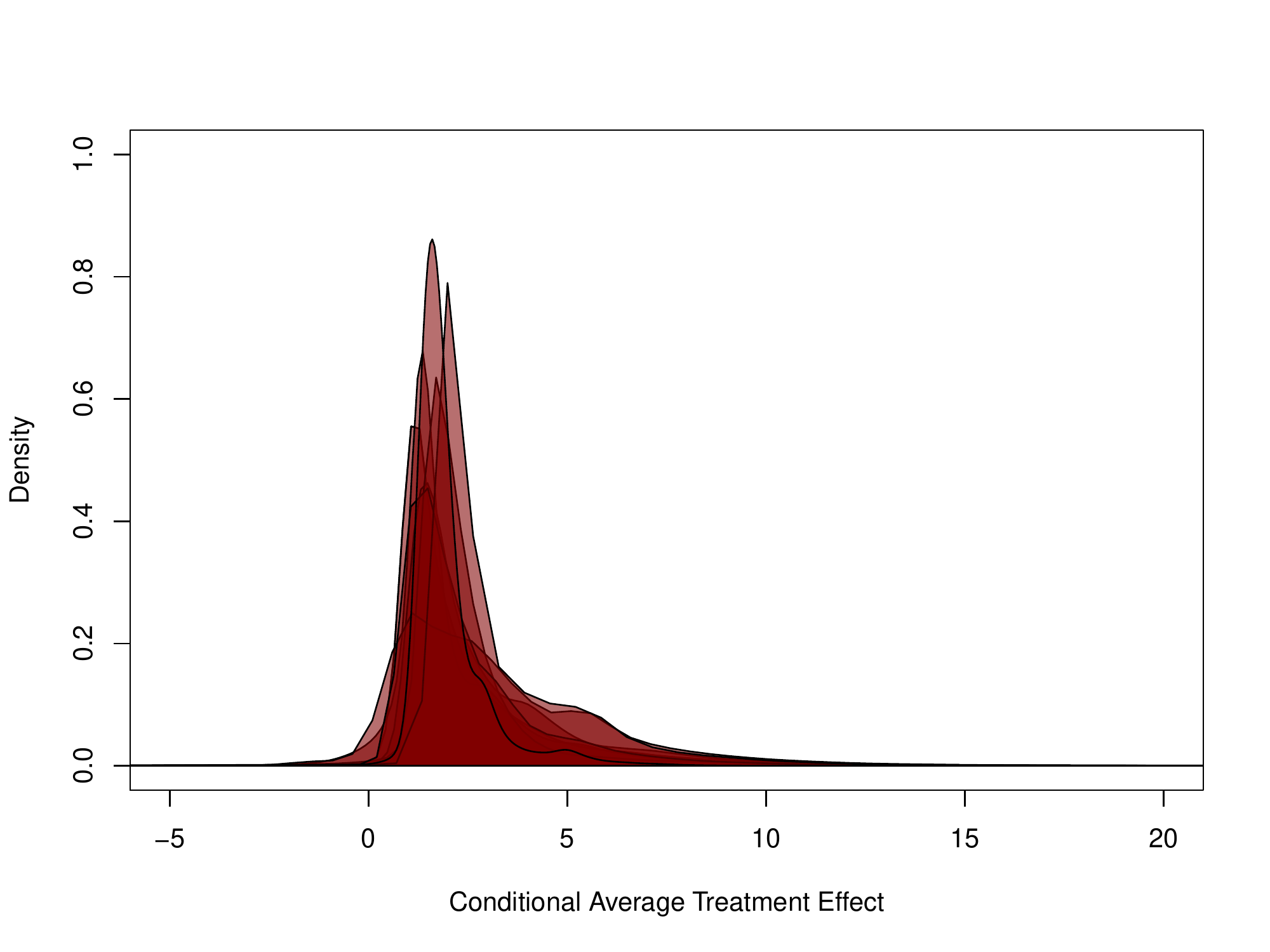}
	\centering\caption{Conditional Average Treatment Effects Across Architectures}
	\label{fig:Conditional-Average-Treatment-1}
\end{figure}

\subsection{Results}

We present now results for treatment effects, utility/profits, and targeting policy evaluations. Table \ref{tab:ATE-Estimates-and} shows the estimates of the average treatment effect from the eight network architectures along with their respective $95\%$ confidence intervals. These results are constructed following Section \ref{sec:inference}, using Equations \eqref{eqn:influence hat} and \eqref{eqn:dr} in particular, and valid by Corollary \ref{cor:rct}. Because this is an experiment, we can compare to the standard unadjusted difference in means, which yields an average treatment effect of $2.561$.

\begin{table}[ht!]
	\caption{Average Treatment Effect Estimates and 95\% Confidence Intervals}
	\label{tab:ATE-Estimates-and}
	\centering 
	\begin{tabular}{cccc}
		\hline
		\hline
		& Average Treatment & & 95\% Confidence   \\ 
		Architecture & Effect ($\hat{\tau}$) & & Interval   \\ 	  
		\hline
		1 & 2.606 & & [2.273 , 2.932]   \\
		2 & 2.577 & & [2.252 , 2.901]   \\    
		3 & 2.547 & & [2.223 , 2.872]   \\    
		4 & 2.488 & & [2.160 , 2.817]   \\    
		5 & 2.459 & & [2.127 , 2.791]   \\    
		6 & 2.430 & & [2.093 , 2.767]   \\    
		7 & 2.400 & & [2.057 , 2.744]   \\    
		8 & 2.371 & & [2.021 , 2.721]   \\     
		\hline 
	\end{tabular}
\end{table}

Turning to expected profits, we estimate $\pi(s) = \E \big[s(\bX)(mY(1) - c) + \left(1 - s(\bX)\right) m Y(0) \big]$, adding a profit margin $m$ and a mailing cost $c$ to \eqref{eqn:profit} (our NDA with the firm forbids revealing $m$ and $c$). We consider three different counterfactual policies $s(\bx)$: (i) \emph{never} treat, $s(\bx) \equiv 0$; (ii) a \emph{blanket} treatment, $s(\bx) \equiv 1$; (iii) a \emph{loyalty} policy, $s(\bx_i) = 1$ only for those who had purchased in the prior calendar year. 
Results are shown in Table \ref{tab:Counterfactual-profits}. It is clear that profits from the three policies are ordered as $\pi(\text{never}) < \pi(\text{blanket}) < \pi(\text{loyalty})$. 

\begin{table}
	\caption{Counterfactual Profits from Three Targeting Strategies}
	\label{tab:Counterfactual-profits}
	\centering
	\begin{tabular}{cccccccccc}
		\hline
		\hline
		& & \multicolumn{2}{c}{\bfseries Never Treat} & & \multicolumn{2}{c}{\bfseries Blanket Treatment} & &  \multicolumn{2}{c}{\bfseries Loyalty Policy} \\
		\cline{3-4} \cline{6-7} \cline{9-10}
		Architecture & & $\hat{\pi}(s)$ & 95\% CI & & $\hat{\pi}(s)$ & 95\% CI & & $\hat{\pi}(s)$ & 95\% CI \\
		\hline 
		1 & & 2.016 & [1.923 , 2.110] & & 2.234 & [2.162 , 2.306] & & 2.367 & [2.292 , 2.443]  \\
		2 & & 2.022 & [1.929 , 2.114] & & 2.229 & [2.157 , 2.301] & & 2.363 & [2.288 , 2.438]  \\
		3 & & 2.027 & [1.934 , 2.120] & & 2.224 & [2.152 , 2.296] & & 2.358 & [2.283 , 2.434]  \\
		4 & & 2.037 & [1.944 , 2.130] & & 2.213 & [2.140 , 2.286] & & 2.350 & [2.274 , 2.425]  \\
		5 & & 2.043 & [1.950 , 2.136] & & 2.208 & [2.135 , 2.281] & & 2.345 & [2.269 , 2.422]  \\
		6 & & 2.048 & [1.954 , 2.142] & & 2.202 & [2.128 , 2.277] & & 2.341 & [2.263 , 2.418]  \\
		7 & & 2.053 & [1.959 , 2.148] & & 2.197 & [2.122 , 2.272] & & 2.336 & [2.258 , 2.414]  \\
		8 & & 2.059 & [1.963 , 2.154] & & 2.192 & [2.116 , 2.268] & & 2.332 & [2.253 , 2.411]  \\     
		\hline
	\end{tabular}
\end{table}

For both the average effects of Table \ref{tab:ATE-Estimates-and} and the counterfactuals of Table \ref{tab:Counterfactual-profits} there is broad agreement among the eight architectures both numerically and substantially. This may be due to the fact that the data is experimental, so that the propensity score is constant. In true observational data this may not be the case. We explore this issue in our Monte Carlo analysis below.

\subsubsection{Placebo Experiment}

We conducted a set of placebo tests to examine whether the deep neural networks we use can truly recover causal effects.
In particular, we 
take only the untreated customers in the data and randomly assign half to treated status.\footnote{We thank Guido Imbens suggesting this analysis.} We then ran the eight architectures of Table \ref{tab:Estimated-Deep-Networks}, as in the true data. The conditional average treatment effects across the architectures
are plotted in Figure \ref{fig:placebo}. We see that the ``true'' zero average effect is recovered precisely and with the expected distribution. The average treatment effect across all models is estimated
to be around -0.024, compared to 2.56 in the original data. Exercises with
different proportions of (placebo) treated customers revealed similar
results.

\begin{figure}
	\includegraphics[scale=0.5]{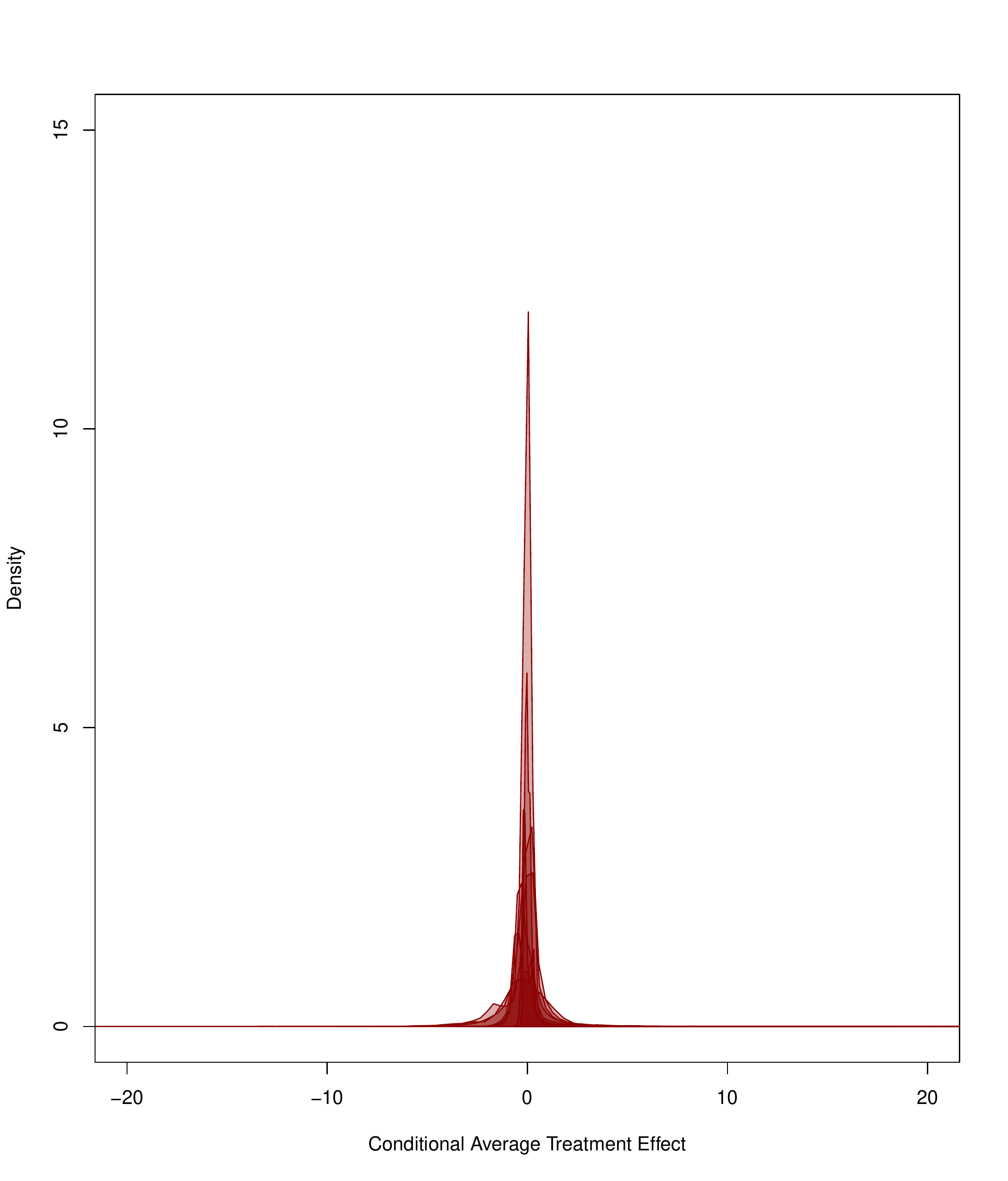}
	\centering\caption{\label{fig:placebo}Placebo Test}
\end{figure}

\subsubsection{Optimal Targeting}

To explore further, we focus on architecture \#3 and study subpopulation treatment targeting strategies following the ideas of Section \ref{sec:policy}. (The other architectures yield similar results, so we omit them.) Architecture \#3 has depth $L=2$ with widths $H_1 = 30$ and $H_2 = 20$. The learning rate was set at 0.0001 and the specification had a total of 4,952 parameters. For this architecture, recalling Remark \ref{rem:regularize}, we added dropout for the second layer with a fixed probability of $1/2$. Using this architecture, we compare the blanket strategy (so $s_0(\bx)=1$) to targeting customers with spend of at least $\bar{y}$ dollars in the prior calendar year (prior spending is one of the covariates), in \$50 increments to \$1200. The policy class is therefore $\mathcal{S} = \{ s(\bx) = \One({\tt prior \ spend} > \bar{y}), \bar{y} = 0, 50, 100, \ldots, 1150, 1200\}$. Figure \ref{fig:Treatment-effects-by} presents the results. The black dots show the difference $\big\{\hat{\pi}(\text{spend}>\bar{y}) - \hat{\pi}(\text{blanket})\big\}$ and the shaded region gives a \emph{pointwise} 95\% confidence band (to ease presentation, sample splitting is not used). We see that there is a significant difference between various choices of $\bar{y}$. Initially, targeting customers with higher spend yields higher profits, as would be expected, but this effect diminishes beyond a certain $\bar{y}$, roughly \$500, as fewer and fewer are targeted. The optimal policy estimate is $\hat{s}(\bx) = \One({\tt prior \ spend} > 400)$. In general, simpler policy classes may yield better decisions, but it is certainly possible to expand our search to different $\mathcal{S}$ by considering further covariates and/or transformations.

\begin{figure}[t]
	\includegraphics[scale=0.7]{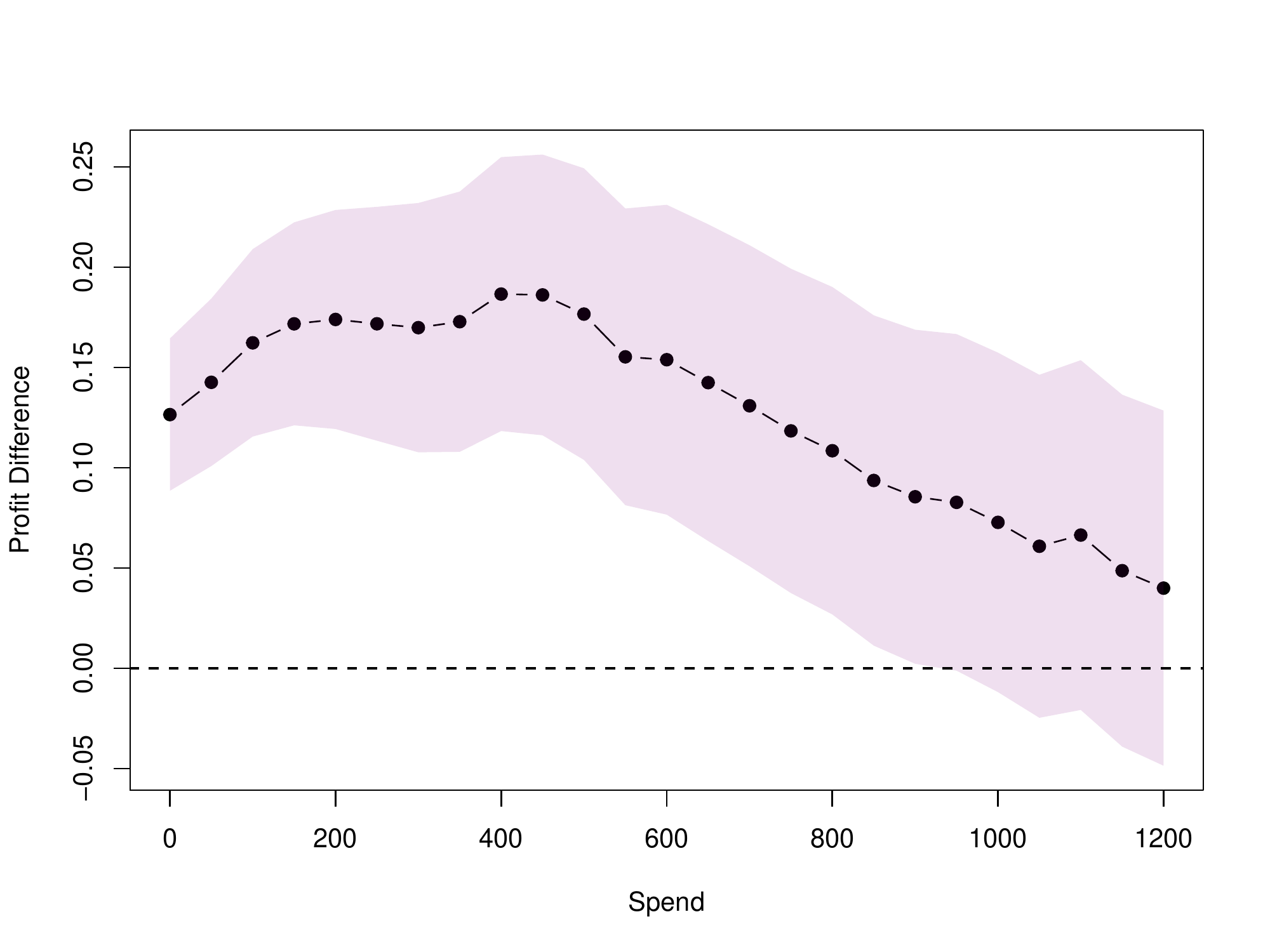}
	\centering\caption{\label{fig:Treatment-effects-by}Expected Profits from Threshold Targeting Based on Prior Year Spend}
\end{figure}

\section{Monte Carlo Analysis}
	\label{sec:simuls}

We conducted a set of Monte Carlo experiments to evaluate our theoretical results. We study inference on the average treatment effect, $\tau$ of \eqref{eqn:ate}, under different data generating processes (DGPs). In each DGP we take $n=10,000$ i.i.d.\ samples and use 1,000 replications. For either $d=20$ or 100, $\bX$ includes a constant term and $d$ independent uniform random variables, $\mathcal{U}(0,1)$. Treatment assignment is Bernoulli with probability $p(\bx)$, where $p(\bx)$ is the propensity score. We consider both (i) randomized treatments with $p(\bx) = 0.5$ and (ii) observational data with $p(\bx)= ( 1+\exp(-\bm{\alpha}_p'\bx) )^{-1}$, where $\alpha_{p,1} = 0.09$ and the remainder are drawn once as $\mathcal{U}(-0.55,0.55)$, and then fixed for the replications. For $d=100$, we maintain $\|\bm{\alpha}_p\|_0 = 20$ for the simplicity. These generate propensities with an approximate range of approximately $(0.30,0.75)$ and mean roughly $0.5$.

Given covariates and treatment assignment, the outcomes are generated according to
\[
	y_i = \mu_0(\bx_i) + \tau(\bx_i) t_i + \varepsilon_i  ,    \qquad         \mu_0(\bx) = \bm{\alpha}_\mu'\bx + \bm{\beta}_{\mu}'\varphi(\bx)  ,    \qquad         \tau\left(x_{i}\right)  = \bm{\alpha}_\tau'\bx + \bm{\beta}_{\tau}'\varphi(\bx),
\]
where $\varepsilon_i \sim \mathcal{N}(0,1)$ and $\varphi(\bx)$ are second-degree polynomials including pairwise interactions. For $\mu_0(\bx)$ and $\tau(\bx)$ we consider two cases, linear and nonlinear models. In both cases the intercepts are $\alpha_{\mu,1} = 0.09$ and $\alpha_{\tau,1} = -0.05$ and slopes are drawn (once) as $\alpha_{\mu,k} \sim \mathcal{N}(0.3,0.7)$ and $\alpha_{\tau,k} \sim \mathcal{U}(0.1,0.22)$, $k = 2, \ldots, d+1$. The linear models set $\bm{\beta}_{\mu} = \bm{\beta}_{\tau} = \bm{0}$ while the nonlinear models take $\beta_{\mu,k} \sim \mathcal{N}(0.01,0.3)$ and $\beta_{\tau,k} \sim \mathcal{U}(-0.05,0.06)$. Altogether, this yields eight designs: $d=20$ or 100, $p(\bx)$ constant or not, and outcome models linear or nonlinear.

For each design, we consider a variety of network architectures, all ReLU-based MLPs. These architectures are variants of the ones used in the empirical application (which were customized for the application). All networks vary in their depth and width, as spelled out in Table \ref{tab:MonteCarlo-Networks}.

\begin{table}
	\caption{\label{tab:MonteCarlo-Networks}Monte Carlo Architectures Explored}
	
	\centering{}%
	\begin{tabular}{cc}
		\hline 
		Architecture & Structure\tabularnewline
		\hline 
		1 & \{20, 15, 5\}\tabularnewline
		2 & \{60, 30, 20\}\tabularnewline
		3 & \{80, 80, 80\}\tabularnewline
		4 & \{20, 15, 10, 5\}\tabularnewline
		5 & \{60, 30, 20, 10\}\tabularnewline
		6 & \{80, 80, 80, 80\}\tabularnewline
		7 & \{20, 15, 15, 10, 10, 5\}\tabularnewline
		8 & \{60, 30, 20, 20, 10, 5\}\tabularnewline
		9 & \{80, 80, 80, 80, 80, 80\}\tabularnewline
		\hline 
	\end{tabular}
\end{table}

Tables \ref{tab:simuls1} and \ref{tab:simuls2} show the results for all eight DGPs. Table \ref{tab:simuls1} shows randomized treatment while Table \ref{tab:simuls2} shows results mimicking observational data. Overall, the results reported show excellent performance of deep learning based semiparametric inference. The bias is minimal and the coverage is quite accurate, while the interval length is under control. Notice that the most architectures yield similar results with no architecture dominating the others. Further, the coverage and interval length are fairly similar with the more complex architecture not exhibiting any systematic patterns of length inflation. 

\begin{table}[h]
	\begin{centering}
		\caption{Simulations Results - Constant Propensity Score\label{tab:simuls1}}
		~
		\par
	\end{centering}
	\begin{onehalfspace}
		\centering{}%
		\begin{tabular}{ccccccccc}
			\multicolumn{9}{c}{}\tabularnewline
			\hline 
			\hline 
			\multirow{2}{*}{Model} & \multirow{2}{*}{Architecture} & \multicolumn{3}{c}{20 Covariates} &  & \multicolumn{3}{c}{100 Covariates}\tabularnewline
			\cline{3-5} \cline{7-9} 
			 &  & Bias & IL & Coverage &  & Bias & IL & Coverage\tabularnewline
			\hline 
			\multirow{9}{*}{\textit{Linear}}  & 1 & 0.00027 & 0.079 & 0.947 &  & 0.00067 & 0.080 & 0.946\tabularnewline
			 & 2 & -0.00032 & 0.079 & 0.951 &  & 0.00012 & 0.080 & 0.958\tabularnewline
			 & 3 & -0.00025 & 0.079 & 0.955 &  & -0.00167 & 0.080 & 0.939\tabularnewline
			 & 4 & -0.00068 & 0.079 & 0.949 &  & 0.00038 & 0.080 & 0.949\tabularnewline
			 & 5 & 0.00008 & 0.079 & 0.945 &  & -0.00219 & 0.080 & 0.929\tabularnewline
			 & 6 & 0.00007 & 0.079 & 0.955 &  & -0.00010 & 0.080 & 0.946\tabularnewline
			 & 7 & 0.00128 & 0.079 & 0.952 &  & -0.00041 & 0.080 & 0.944\tabularnewline
			 & 8 & 0.00108 & 0.079 & 0.949 &  & -0.00088 & 0.080 & 0.941\tabularnewline
			 & 9 & 0.00021 & 0.078 & 0.948 &  & -0.00080 & 0.081 & 0.953\tabularnewline
			\hline 
			\multirow{9}{*}{\textit{Nonlinear}}  & 1 & 0.00087 & 0.081 & 0.946 &  & -0.00067 & 0.163 & 0.940\tabularnewline
			 & 2 & 0.00015 & 0.079 & 0.954 &  & 0.00093 & 0.153 & 0.927\tabularnewline
			 & 3 & -0.00072 & 0.079 & 0.940 &  & 0.00245 & 0.148 & 0.926\tabularnewline
			 & 4 & 0.00101 & 0.080 & 0.945 &  & -0.00087 & 0.165 & 0.956\tabularnewline
			 & 5 & 0.00027 & 0.079 & 0.935 &  & -0.00190 & 0.154 & 0.923\tabularnewline
			 & 6 & -0.00025 & 0.079 & 0.929 &  & -0.00117 & 0.146 & 0.902\tabularnewline
			 & 7 & -0.00052 & 0.080 & 0.947 &  & 0.00091 & 0.165 & 0.941\tabularnewline
			 & 8 & 0.00077 & 0.079 & 0.938 &  & 0.00201 & 0.153 & 0.927\tabularnewline
			 & 9 & -0.00013 & 0.079 & 0.940 &  & 0.00049 & 0.154 & 0.936\tabularnewline
			\hline 
			\hline 
			\multicolumn{9}{c}{}\tabularnewline
		\end{tabular}
	\end{onehalfspace}
\end{table}

\begin{table}[h]
	\begin{centering}
	\caption{Simulations Results - Non-constant Propensity Score \label{tab:simuls2}}
	~
	\par\end{centering}
	\begin{onehalfspace}
		\centering{}%
		\begin{tabular}{ccccccccc}
			\multicolumn{9}{c}{}\tabularnewline
			\hline 
			\hline 
			\multirow{2}{*}{Model} & \multirow{2}{*}{Architecture} & \multicolumn{3}{c}{20 Covariates} &  & \multicolumn{3}{c}{100 Covariates}\tabularnewline
			\cline{3-5} \cline{7-9} 
			 &  & Bias & IL & Coverage &  & Bias & IL & Coverage\tabularnewline
			\hline 
			\multirow{9}{*}{\textit{Linear}}  & 1 & -0.00202 & 0.080 & 0.948 &  & 0.0009 & 0.081 & 0.955\tabularnewline
			 & 2 & 0.00011 & 0.079 & 0.946 &  & 0.0007 & 0.081 & 0.945\tabularnewline
			 & 3 & -0.00130 & 0.079 & 0.964 &  & -0.0001 & 0.081 & 0.937\tabularnewline
			 & 4 & -0.00106 & 0.079 & 0.945 &  & 0.0002 & 0.081 & 0.933\tabularnewline
			 & 5 & -0.00083 & 0.079 & 0.951 &  & -0.0004 & 0.081 & 0.944\tabularnewline
			 & 6 & -0.00068 & 0.079 & 0.955 &  & 0.0001 & 0.081 & 0.924\tabularnewline
			 & 7 & -0.00119 & 0.079 & 0.953 &  & -0.0001 & 0.081 & 0.942\tabularnewline
			 & 8 & -0.00056 & 0.079 & 0.952 &  & -0.0008 & 0.081 & 0.939\tabularnewline
			 & 9 & -0.00096 & 0.079 & 0.948 &  & -0.0007 & 0.081 & 0.952\tabularnewline
			\hline 
			\multirow{9}{*}{\textit{Linear}}  & 1 & -0.00076 & 0.081 & 0.946 &  & -0.00279 & 0.164 & 0.937\tabularnewline
			 & 2 & -0.00122 & 0.080 & 0.939 &  & 0.00020 & 0.155 & 0.941\tabularnewline
			 & 3 & -0.00074 & 0.080 & 0.926 &  & -0.00080 & 0.148 & 0.914\tabularnewline
			 & 4 & -0.00171 & 0.081 & 0.940 &  & -0.00184 & 0.166 & 0.938\tabularnewline
			 & 5 & -0.00135 & 0.080 & 0.952 &  & -0.00103 & 0.154 & 0.912\tabularnewline
			 & 6 & -0.00075 & 0.080 & 0.950 &  & -0.00174 & 0.147 & 0.905\tabularnewline
			 & 7 & -0.00153 & 0.081 & 0.928 &  & -0.00377 & 0.165 & 0.929\tabularnewline
			 & 8 & 0.00082 & 0.080 & 0.953 &  & 0.00031 & 0.154 & 0.919\tabularnewline
			 & 9 & -0.00127 & 0.080 & 0.931 &  & -0.00094 & 0.156 & 0.917\tabularnewline
			\hline 
			\hline 
			\multicolumn{9}{c}{}\tabularnewline
		\end{tabular}
	\end{onehalfspace}
\end{table}

None of the architectures we presented earlier used regularization. In typical empirical applications, including our own, researchers adopt architectures that employ dropout, a common method of regularization; see Remark \ref{rem:regularize}. Our own preliminary exploration of dropout and other forms of regularization found expected departures form nonregularized models. In most, but not all, cases the coverage remained accurate, but with increased bias and interval length compared to Table \ref{tab:simuls1} and \ref{tab:simuls2}. The results preach caution when applying regularization in applications.

\section{Conclusion}
	\label{sec:conclusion}

The utility of deep learning in social science applications is still a subject of interest and debate. While there is an acknowledgment of its predictive power, there has been limited adoption of deep learning in social sciences such as economics. Some part of the reluctance to adopting these methods stems from the lack of theory facilitating use and interpretation. We have shown, both theoretically as well as empirically, that these methods can offer excellent performance.

In this paper, we have given a formal proof that inference can be valid after using deep learning methods for first-step estimation. To the best of our knowledge, ours is the first inference result using deep nets. Our results thus contribute directly to the recent explosion in both theoretical and applied research using machine learning methods in economics, and to the recent adoption of deep learning in empirical settings. We obtained novel bounds for deep neural networks, speaking directly to the modern (and empirically successful) practice of using fully-connected feedfoward networks. Our results allow for different network architectures, including fixed width, very deep networks. Our results cover general nonparametric regression-type loss functions, covering most nonparametric practice. We used our bounds to deliver fast convergence rates allowing for second-stage inference on a finite-dimensional parameter of interest.

There are practical implications of the theory presented in this paper. We focused on semiparametric causal effects as a concrete illustration, but deep learning is a potentially valuable tool in many diverse economic settings. Our results allow researchers to embed deep learning into standard econometric models such as linear regressions, generalized linear models, and other forms of limited dependent variables models (e.g. censored regression). Our theory can also be used as a starting point for constructing deep learning implementations of two-step estimators in the context of selection models, dynamic discrete choice, and the estimation of games.

To be clear, we see our paper as a first step in the exploration of deep learning as a tool for economic applications. There are a number of opportunities, questions, and challenges that remain. For example, factor models in finance might benefit from the use of auto-encoders and recurrent neural nets may have applications in time series. For some estimands, it may be crucial to estimate the density as well, and this problem can be challenging in high dimensions. Deep nets, in the form of GANs are a promising tool for distribution estimation. There are also interesting questions remaining as to an optimal network architecture, and if this can be itself learned from the data, as well as computational and optimization guidance. Research into these further applications and structures is underway.

\section{References}
\singlespacing
\begingroup
\renewcommand{\section}[2]{}	
\bibliography{Farrell-Liang-Misra2019_deepNetsATE--Bibliography}{}
\bibliographystyle{ecta}
\endgroup

\onehalfspacing
\begin{appendices}

\section{Proofs}
	\label{appx:proof}

In this section we provide a proof of Theorems \ref{thm:mlp rate} and \ref{thm:generic rates}, our main theoretical results for deep ReLU networks, and their corollaries. The proof proceeds in several steps. We first give the main breakdown and bound the bias (approximation error) term. We then turn our attention to the empirical process term, to which we apply our localization. Much of the proof uses a generic architecture, and thus pertains to both results. We will specialize the architecture to the multi-layer perceptron only when needed later on. Other special cases and related results are covered in Section \ref{appx:coro}. Supporting Lemmas are stated in Section \ref{appx:lemmas}. 

The statements of Theorems \ref{thm:mlp rate} and \ref{thm:generic rates} assume that $n$ is large enough. Precisely, we require $n > (2eM)^2  \vee {\rm Pdim}(\cF_{\rm DNN})$. For notational simplicity we will denote $\widehat{f}_{\rm DNN} \deq \hat{f}$, see \eqref{eqn:erm}, and $\epsilon_{\rm DNN} \deq \epsilon_n$, see Assumption \ref{asmpt:bias}. As we are simultaneously consider Theorems \ref{thm:mlp rate} and \ref{thm:generic rates}, the generic notation ${\rm DNN}$ will be used throughout.

\subsection{Main Decomposition and Bias Term}
	\label{appx:main-steps}
	
	Referring to Assumption \ref{asmpt:bias}, define the best approximation realized by the deep ReLU network class $\cF_{\rm DNN}$ as 
	\[
		f_n \deq \argmin_{\substack{f \in \cF_{\rm DNN} \\ \| f \|_\infty \leq 2M}} \| f - f_* \|_{\infty}.
	\]
	By definition, $\epsilon_n \deq \epsilon_{\rm DNN} \deq \| f_n - f_* \|_\infty$.

	Recalling the optimality of the estimator in \eqref{eqn:erm}, we know, as both $f_n$ and $\hat{f}$ are in $\cF_{\rm DNN}$, that
	\[
		 - \E_n [\ell(\hat{f}, \bz)] + \E_n[ \ell(f_n, \bz)]  \geq 0.
	\]
	This result does not hold for $f_*$ in place of $f_n$, because $f_* \not\in \cF_{\rm DNN}$. Using the above display and the curvature of Equation \eqref{eqn:loss} (which does not hold with $f_n$ in place of $f_*$ therein), we obtain
	\begin{align}
		c_1 \| \hat{f} - f_* \|_{L_2(\bx)}^2 & \leq \E [\ell(\hat{f}, \bz)] - \E [\ell(f_*, \bz)] \nonumber \\
		& \leq \E [\ell(\hat{f}, \bz)] - \E [\ell(f_*, \bz)] - \E_n [\ell(\hat{f}, \bz)] + \E_n [\ell(f_n, \bz)] \nonumber \\
		& = \E \left[ \ell(\hat{f}, \bz) - \ell(f_*, \bz) \right] - \E_n \left[\ell(\hat{f}, \bz) -  \ell(f_*, \bz) \right] + \E_n \left[ \ell(f_n, \bz) - \ell(f_*, \bz) \right] \nonumber \\
		& = \left( \E - \E_n \right)\left[ \ell(\hat{f}, \bz) - \ell(f_*, \bz) \right] + \E_n \left[ \ell(f_n, \bz) - \ell(f_*, \bz) \right]. \label{eqn:decomposition}
	\end{align}

	Equation \eqref{eqn:decomposition} is the main decomposition that begins the proof. The decomposition must be done this way because of the above notes regarding $f_*$ and $f_n$.  The first term is the empirical process term that will be treated in the subsequent subsection. For the second term in \eqref{eqn:decomposition}, the bias term or approximation error, we apply Bernstein's inequality to find that, with probability at least $1 - e^{-\gt}$,
	\begin{align}
		\E_n \left[ \ell(f_n, \bz) - \ell(f_*, \bz) \right] &\leq \E \left[ \ell(f_n, \bz) - \ell(f_*, \bz) \right] + \sqrt{\frac{2C_\ell^2 \| f_n - f_*\|_\infty^2 \gt }{n}} + \frac{21C_\ell M \gt}{3n}   		\nonumber \\
		& \leq c_2 \E \left[ \| f_n - f_* \|^2 \right] + \sqrt{\frac{2C_\ell^2 \| f_n - f_*\|_\infty^2 \gt }{n}} + \frac{7 C_\ell M \gt}{n}   		\nonumber \\
		& \leq c_2 \epsilon_n^2 + \epsilon_n \sqrt{\frac{2C_\ell^2 \gt }{n}} + \frac{7 C_\ell M \gt}{n},   		\label{eqn:bias}
	\end{align}
	using the Lipschitz and curvature of the loss function defined in Equation \eqref{eqn:loss} and $\E \left[\| f_n - f_* \|^2 \right]\leq \| f_n - f_*\|_\infty^2$, along with the definition of $\epsilon_n^2$.

	Once the empirical process term is controlled (in Section \ref{appx:localization}), the two bounds will be brought back together to compute the final result, see Section \ref{appx:mlp-last-steps}.

\subsection{Localization Analysis}
	\label{appx:localization}

We now turn to bounding the first term in \eqref{eqn:decomposition} (the empirical processes term) using a localized analysis that derives bounds based on scale insensitive complexity measure. The ideas of our localization are rooted in \cite{koltchinskii2000rademacher} and \cite{bartlett2005local}, and related to \cite{Koltchinskii2011_book}. Localization analysis extending to the unbounded $f$ case has been developed in \cite{mendelson2014learning, liang2015learning}. This proof section proceeds in several steps. 
	
	
A key quantity is the Rademacher complexity of the function class at hand. Given i.i.d.\ Rademacher draws, $\eta_i = \pm 1$ with equal probability independent of the data, the random variable $R_n \cF$, for a function class $\cF$, is defined as
\[
	R_n \cF : = \sup_{f \in \cF} \frac{1}{n}\sum_{i=1}^n \eta_i f(\bx_i).
\]
Intuitively, $R_n \cF$ measures how flexible the function class is for predicting random signs. Taking the expectation of $R_n \cF$ conditioned on the data we obtain the \emph{empirical Rademacher complexity}, denoted $\E_\eta [R_n \cF]$. When the expectation is taken over both the data and the draws $\eta_i$, $\E R_n \cF$, we get the \emph{Rademacher complexity}.

\subsubsection{Step I: Quadratic Process} 
	\label{appx:quadratic}

The first step is to show that, with high probability, the empirical $L_2$ norm of the error $(f - f_*)$ is at most twice the population $L_2$ norm bound for the same error, for certain functions $f$ outside a certain critical radius. This will be an ingredient to be used later on. Denote $\| f \|_n := \left( \frac{1}{n} \sum_{i=1}^n f(\bx_i)^2 \right)^{1/2}$ to be the empirical $L_2$ norm. To do so, we study the quadratic process
\begin{align*}
	\| f - f_* \|_n^2 - \| f - f_* \|_{L_2(\bx)}^2 = \E_n (f - f_* )^2 - \E (f - f_* )^2.
\end{align*}

We will apply the symmetrization of Lemma \ref{lem:symmetrization} to $g = ( f - f_* )^2$ restricted to a radius $\|f - f_* \|_{L_2(\bx)} \leq r$. This function $g$ has variance bounded as 
\[
	\V[g]  \leq \E[g^2] \leq \E (( f - f_* )^4) \leq 9M^2 r^2.
\]
Writing $g = ( f + f_* ) ( f - f_* )$, we see that by Assumption \ref{asmpt:dgp}, $|g| \leq  3M | f - f_*| \leq 9M^2$, where the first inequality verifies that $g$ has a Lipschitz constant of $3M$ (when viewed as a function of its argument $f$), and second that $g$ itself is bounded. We therefore apply Lemma \ref{lem:symmetrization}, to obtain, with probability at least $1 - \exp(-\gt)$, that for any $f \in \cF$ with $\|f - f_* \|_{L_2(\bx)} \leq r$,
\begin{align}
	& \E_n (f - f_* )^2  - \E (f - f_* )^2    		\nonumber \\
	& \qquad  \leq 3 \E R_n \{ g = ( f - f_* )^2: f \in \cF, \|f - f_* \|_{L_2(\bx)} \leq r \} + 3 M r \sqrt{\frac{2 \gt}{n}} + \frac{36 M^2}{3} \frac{\gt}{n}    		\nonumber \\
	& \qquad  \leq  18 M \E R_n \{ f - f_*: f \in \cF, \|f - f_* \|_{L_2(\bx)} \leq r \} + 3Mr \sqrt{\frac{2 \gt}{n}} + \frac{12 M^2 \gt}{n},      \label{eqn:quad-base}
\end{align}
where the second inequality applies Lemma \ref{lem:contraction} to the Lipschitz functions $\{g\}$ (as a function of the real values $f(\bx)$) and iterated expectations.

Suppose the radius $r$ satisfies 
\begin{align}
	r^2 &\geq 18M \E R_n \{ f - f_*: f \in \cF, \|f - f_* \|_{L_2(\bx)} \leq r \}  \label{eqn:cr1}
\intertext{and}
	r^2 &\geq \frac{6\sqrt{6} M^2 \gt}{n}.     \label{eqn:cr2}
\end{align}  
Then we conclude from from \eqref{eqn:quad-base} that
\begin{align}
	\E_n (f - f_* )^2 & \leq r^2 + r^2 + 3Mr \sqrt{\frac{2 \gt}{n}} + \frac{12 M^2\gt}{n}   \leq (2r)^2     \label{eqn:radius-factor-2}
\end{align}
where the first inequality uses \eqref{eqn:cr1} and the second line uses \eqref{eqn:cr2}. This means that for $r$ above the ``critical radius'' (see {\bf Step III}), the empirical $L_2$-norm is at most twice the population one with probability at least $1 - \exp(-\gt)$.

\subsubsection{Step II: One Step Improvement} 
	\label{appx:one-step}

In this step we will show that given a bound on $\| \hat{f} - f_* \|_{L_2(\bx)}$ we can use this bound as information to obtain a tighter bound, if the initial bound is loose as made precise at the end of this step. This tightening will then be pursued to its limit in {\bf Step III}, which leads to the final rate obtained in {\bf Step IV}. {\bf Step I} will be used herein. 

Suppose we know that for some $r_0$, $\| \hat{f} - f_* \|_{L_2(\bx)} \leq r_0$. We may always start with $r_0 = 3M$ given Assumption \ref{asmpt:dgp} and \eqref{eqn:erm}. Apply Lemma \ref{lem:symmetrization} with $\cG \deq \{g = \ell(f, \bz) - \ell(f_*, \bz) : f \in  \cF_{\rm DNN}, \| f - f_* \|_{L_2(\bx)} \leq r_0 \}$, we find that, with probability at least $1-2e^{-\gt}$, the empirical process term of \eqref{eqn:decomposition} is bounded as
\begin{equation}
	\label{eqn:apply-sym-one-step}
	\left( \E - \E_n \right)\left[ \ell(\hat{f}, \bz) - \ell(f_*, \bz) \right]
	\leq 6 \E_\eta R_n \cG  + \sqrt{\frac{2 C_\ell^2 r_0^2 \gt}{n}} + \frac{23 \cdot 3 M C_\ell}{3} \frac{\gt}{n},
\end{equation}
where the middle term is due to the following variance calculation (recall Equation \eqref{eqn:loss})
\begin{align*}
	\V[g] & \leq \E[g^2] = \E[  |\ell(f, \bz) - \ell(f_*, \bz) |^2] \leq C_\ell^2 \E (f - f_*)^2 \leq C_\ell^2 r_0^2 
\end{align*}
Here the fact that Lemma \ref{lem:symmetrization} is variance dependent, and that the variance depends on the radius $r_0$, is important. It is this property which enables a sharpening of the rate with step-by-step reductions in the variance bound, as in Section \ref{appx:shells}.

For the empirical Rademacher complexity term, the first term of \eqref{eqn:apply-sym-one-step}, Lemma \ref{lem:contraction}, {\bf Step I}, and Lemma \ref{lem:dudley-chaining}, yield
\begin{align*}
	 \E_\eta R_n \cG &=\E_\eta R_n \{g : g = \ell(f, \bz) - \ell(f_*, \bz), f \in \cF_{\rm DNN}, \| f -f_* \| \leq r_0\} \\
	& \leq 2 C_\ell \E_\eta R_n \{f - f_* :  f \in \cF_{\rm DNN}, \| f -f_* \| \leq r_0 \} \\
	& \leq 2 C_\ell \E_\eta R_n \{f - f_* :  f \in \cF_{\rm DNN}, \| f -f_* \|_n \leq 2 r_0 \} \\
	& \leq 2 C_\ell \inf_{0 < \alpha < 2r_0} \left\{ 4 \alpha + \frac{12}{\sqrt{n}} \int_\alpha^{2r_0} \sqrt{\log \cN(\delta, \cF_{\rm DNN}, \| \cdot\|_n) } d\delta  \right\}     \\
	& \leq 2 C_\ell \inf_{0 < \alpha < 2r_0} \left\{ 4 \alpha + \frac{12}{\sqrt{n}} \int_\alpha^{2 r_0} \sqrt{\log \cN(\delta, \cF_{\rm DNN}|_{x_1, \ldots, x_n}, \infty)  } d\delta  \right\} \enspace,
\end{align*}
with probability $1 - \exp(-\gt)$ (when applying {\bf Step I}). Recall Lemma \ref{lem:ps-dim}, one can further upper bound the entropy integral when $n > {\rm Pdim}(\cF_{\rm DNN})$,
\begin{align*}
	& \inf_{0 < \alpha < 2r_0} \left\{ 4 \alpha + \frac{12}{\sqrt{n}} \int_\alpha^{2 r_0} \sqrt{\log \cN(\delta, \cF_{\rm DNN}|_{x_1, \ldots, x_n}, \infty)  } d\delta  \right\} \nonumber \\
	& \leq \inf_{0 < \alpha < 2r_0} \left\{ 4 \alpha + \frac{12}{\sqrt{n}} \int_\alpha^{2 r_0} \sqrt{ {\rm Pdim}(\cF_{\rm DNN}) \log \frac{2eM n}{\delta \cdot {\rm Pdim}(\cF_{\rm DNN})}  } d\delta  \right\} \nonumber \\
	& \leq 32 r_0 \sqrt{\frac{{\rm Pdim}(\cF_{\rm DNN})}{n} \left( \log \frac{2eM}{r_0} + \frac{3}{2} \log n \right) } \label{eqn:chaining-bd}
\end{align*}
with a particular choice of $\alpha = 2r_0 \sqrt{{\rm Pdim}(\cF_{\rm DNN}) / n} < 2r_0$. Therefore, whenever $r_0 \geq 1/n$ and $n \geq (2eM)^2$, 
\[
	 \E_\eta R_n \cG  \leq  128 C_\ell r_0  \sqrt{ \frac{{\rm Pdim}(\cF_{\rm DNN}) }{n}  \log n}.
\]
Applying this bound to \eqref{eqn:apply-sym-one-step}, we have
\begin{equation}
	\label{eqn:emp-term}
	\left( \E - \E_n \right)\left[ \ell(\hat{f}, \bz) - \ell(f_*, \bz) \right]
	\leq K r_0  \sqrt{ \frac{{\rm Pdim}(\cF_{\rm DNN})}{n} \log n} + r_0 \sqrt{\frac{2 C_\ell^2 \gt}{n}} + \frac{23 M C_\ell \gt}{n}  
\end{equation}
where $K = 6\times128C_\ell$.

Going back now to the main decomposition, plug \eqref{eqn:emp-term} and \eqref{eqn:bias} into \eqref{eqn:decomposition}, and we overall have found that, with probability at least $1- 4\exp(-\gt)$, the following holds:
\begin{align}
	& c_1 \| \hat{f} - f_* \|_{L_2(\bx)}^2  		 \nonumber \\
	& \leq K r_0  \sqrt{ \frac{{\rm Pdim}(\cF_{\rm DNN})}{n} \log n} + r_0 \sqrt{\frac{2 C_\ell^2 \gt}{n}} + \frac{23 M C_\ell \gt}{n} + \left( c_2 \epsilon_n^2 + \epsilon_n \sqrt{\frac{ 2C_\ell^2 \gt }{n}} + \frac{7 C_\ell M \gt}{n} \right)   		\nonumber \\
	& \leq r_0 \cdot   \left( K\sqrt{ \frac{{\rm Pdim}(\cF_{\rm DNN})}{n} \log n} + \sqrt{\frac{2 C_\ell^2 \gt}{n}} \right) + c_2 \epsilon_n^2 + \epsilon_n \sqrt{\frac{2C_\ell^2 \gt }{n}}   + 30 M C_\ell \frac{\gt}{n}   		\nonumber \\
	& \leq r_0 \cdot \left( K\sqrt{C} \sqrt{ \frac{W L \log W}{n}\log n}  + \sqrt{\frac{2 C_\ell^2 \gt}{n}} \right) + c_2 \epsilon_n^2 + \epsilon_n \sqrt{\frac{2C_\ell^2 \gt }{n}}   + 30 M C_\ell \frac{\gt}{n}.  		\label{eqn:one-step}
\end{align}
The last line applies Lemma \ref{lem:VC}. Therefore, whenever $\epsilon_n \ll r_0$ and $ \sqrt{ \frac{W L \log W}{n}\log n} \ll r_0$, the knowledge that $\|\hat{f} - f_*\|_{L_2(\bx)} \leq r_0$ implies that (with high probability) $\| \hat{f} - f_* \|_{L_2(\bx)} \leq r_1$, for $r_1 \ll r_0$. One can recursively improve the bound $r$ to a fixed point/radius $r_*$, which describes the fundamental difficulty of the problem. This is done in the course of the next two steps.

\subsubsection{Step III: Critical Radius} 
	\label{appx:radius}

We now use the tightening of {\bf Step II} to obtain the critical radius for this problem that is then used as an input in the final rate derivation of {\bf Step IV}. Formally, define the critical radius $r_*$ to be the largest fixed point 
\[
	r_* = \inf \left\{ r>0 : 18 M \E R_n \{ f - f_*: f \in \cF, \|f - f_* \|_{L_2(\bx)} \leq s \} < s^2, \forall s \geq r \right\}.
\]
By construction this obeys \eqref{eqn:cr1}, and thus so does $2r_*$. Denote the event $E$ (depending on the data) to be 
\begin{align*}
	E = \left\{ \|f - f_* \|_n \leq 4r_*,~~\text{for all $f \in \cF$ and}~~ \|f - f_* \|_{L_2(\bx)} \leq 2r_* \right\}
\end{align*}
and $\One_E$ to be the indicator that event $E$ holds. We know from \eqref{eqn:radius-factor-2} that $\P(\One_E=1) \geq 1 - n^{-1}$, provided $r_* \geq \sqrt{18}M\sqrt{\log n/n}$ to satisfy \eqref{eqn:cr2}.

We can now give an upper bound for the the critical radius $r_*$. Using the logic of {\bf Step II} to bound the empirical Rademacher complexity, and then applying Lemma \ref{lem:VC}, we find that
\begin{align*}
	r_*^2 &\leq 18 M \E R_n \big\{ f - f_*: f \in \cF, \|f - f_* \|_{L_2(\bx)} \leq r_* \big\}    		\\
	& \leq 18 M \E R_n \big\{ f - f_*: f \in \cF, \|f - f_* \|_{L_2(\bx)} \leq 2r_* \big\}    		\\
	& \leq 18 M \E \big\{ \E_\eta R_n \{ f - f_*: f \in \cF, \|f - f_* \|_{n} \leq 4r_* \} \One_E + 3 M (1-\One_E)  \big\}     \\ 
	& \leq 36 M K \sqrt{C} \cdot r_*  \sqrt{ \frac{W L \log W}{n}\log n} + 36 M^2 \frac{1}{n}  		\\
	& \leq 72 M K \sqrt{C} \cdot r_*  \sqrt{ \frac{W L \log W}{n}\log n},
\end{align*}
with the last line relying on the above restriction that $r_* \geq \sqrt{18}M\sqrt{\log n/n}$. Dividing through by $r_*$ yields the final bound:
\begin{align}
	\label{eqn:r_star}
	r_* & \leq 72 M K\sqrt{C}  \sqrt{ \frac{W L \log W}{n}\log n}.
\end{align}

\subsubsection{Step IV: Localization}
	\label{appx:shells}

We are now able to derive the final rate using a localization argument. This applies the results of {\bf Step I} and {\bf Step II} repeatedly. Divide the space $\cF_{\rm DNN}$ into shells of increasing radius by intersecting it with the $L_2$ balls
\begin{align}
	B(f_*, \bar{r}), B(f_*, 2\bar{r}) \backslash B(f_*, \bar{r}), \ldots B(f_*, 2^{l} \bar{r}) \backslash B(f_*, 2^{l-1} \bar{r})
\end{align}
where $l \geq 1$ is chosen to be the largest integer no greater than $\log_2 \frac{2M}{\sqrt{(\log n)/n}}$. We will proceed to find a bound on $\bar{r}$ which determines the final rate results.

Suppose $\bar{r} > r_*$. Then for each shell, {\bf Step I} and the union bound imply that with probability at least $1 - 2 l \exp(-\gt)$, 
\begin{align}
	\| f - f_* \|_{L_2(\bx)} \leq 2^{j} \bar{r}  \ \Rightarrow  \ \| f - f_* \|_n \leq 2^{j+1} \bar{r} .
\end{align}

Further, suppose that for some $j \leq l$
\begin{align}
	\hat{f} \in B(f_*, 2^{j} \bar{r}) \backslash B(f_*, 2^{j-1} \bar{r}). 
\end{align}
Then applying the one step improvement argument in {\bf Step II} (again the variance dependence captured in Lemma \ref{lem:symmetrization} is crucial, here reflected in the variance within each shell), Equation \eqref{eqn:one-step} yields that with probability at least $1 - 4 \exp(-\gt)$,
\begin{align*}
	\| \hat{f} - f_* \|_{L_2(\bx)}^2 &\leq  \frac{1}{c_1} \left\{ 2^j \bar{r} \cdot \left( K\sqrt{C} \sqrt{ \frac{W L \log W}{n}\log n}  + \sqrt{\frac{2 C_\ell^2 t}{n}} \right) + c_2 \epsilon_n^2 + \epsilon_n \sqrt{\frac{2C_\ell^2 \gt }{n}}   + 30 M C_\ell \frac{\gt}{n} \right\} \\
	& \leq 2^{2j-2} \bar{r}^2,
\end{align*}
if the following two conditions hold:
\begin{align*}
	  \frac{1}{c_1}  \left( K\sqrt{C} \sqrt{ \frac{W L \log W}{n}\log n} + \sqrt{\frac{2 C_\ell^2 \gt}{n}} \right)  \leq \frac{1}{2} 2^{j}\bar{r}   	\\
	  \frac{1}{c_1}\left( c_2 \epsilon_n^2 + \epsilon_n \sqrt{\frac{2C_\ell^2 \gt }{n}}   + 26 MC_\ell \frac{\gt}{n}  \right) \leq \frac{1}{4} 2^{2j}\bar{r}^2.
\end{align*}

It is easy to see that these two hold for all $j$ if we choose 
\begin{align}
	\label{eqn:r_bar}
	\bar{r} = \frac{8}{c_1}  \left( K\sqrt{C} \sqrt{ \frac{W L \log W}{n}\log n}  + \sqrt{\frac{2 C_\ell^2 \gt}{n}} \right) + \left( \sqrt{\frac{2 (c_2 \vee 1)}{c_1}} \epsilon_n + \sqrt{\frac{120 M C_\ell}{c_1} \frac{\gt}{n}} \right) + r_* .
\end{align}

Therefore with probability at least $1 - 6l \exp(-\gt)$, we can perform shell-by-shell argument combining the results in {\bf Step I} and {\bf Step II}:
\begin{align*}
	& \| \hat{f} - f_* \|_{L_2(\bx)} \leq 2^{l} \bar{r} ~~\text{and}~~ \| \hat{f} - f_* \|_n \leq 2^{l+1} \bar{r} \\
	& \qquad \text{implies } \ \| \hat{f} - f_* \|_{L_2(\bx)} \leq 2^{l-1} \bar{r} ~~\text{and}~~ \| \hat{f} - f_* \|_n \leq 2^{l}  \bar{r} \\
	& \qquad\quad \ldots \ldots \\
	& \qquad\qquad  \text{implies } \ \| \hat{f} - f_* \|_{L_2(\bx)} \leq  2^0 \bar{r} ~~\text{and}~~ \| \hat{f} - f_* \|_n \leq 2^1 \bar{r}.
\end{align*}
The ``and'' part of each line follows from {\bf Step I} and the implication uses the above argument following {\bf Step II}. Therefore in the end, we conclude with probability at least $1 - 6l \exp(-\gt)$, 
\begin{align}
\| \hat{f} - f_* \|_{L_2(\bx)} \leq  \bar{r} \enspace, \\
\| \hat{f} - f_* \|_n \leq 2 \bar{r} \enspace.	
\end{align}

Therefore choose $\gamma = - \log (6l) + \gt$, we know from \eqref{eqn:r_bar}, and the upper bound on $r_*$ in \eqref{eqn:r_star}
\begin{align}
	\bar{r} &\leq \frac{8}{c_1}  \left( K\sqrt{C} \sqrt{ \frac{W L \log W}{n}\log n}  + \sqrt{\frac{2 C_\ell^2 (\log \log n + \gamma)}{n}} \right)    		\nonumber \\
	& \qquad\qquad\qquad\qquad\qquad + \left( \sqrt{\frac{2(c_2 \vee 1)}{c_2}} \epsilon_n + \sqrt{\frac{120 M C_\ell}{c_1} \frac{\log \log n + \gamma}{n}} \right) + r_* \nonumber \\
	& \leq C' \left( \sqrt{ \frac{W L \log W}{n}\log n} + \sqrt{\frac{\log \log n + \gamma}{n}} + \epsilon_n  \right),  		\label{eqn:balance}
\end{align} 
with some constant $C' > 0$ that does not depend on $n$. This completes the proof of Theorem \ref{thm:generic rates}.

\subsection{Final Steps for the MLP case}
	\label{appx:mlp-last-steps}

For the multi-layer perceptron, $W \leq C \cdot H^2 L$, and plugging this into the bound \eqref{eqn:balance}, we obtain
\[
	C' \left( \sqrt{ \frac{H^2 L^2 \log (H^2 L)}{n}\log n} + \sqrt{\frac{\log \log n + \gamma}{n}} + \epsilon_n  \right)
\]
	
To optimize this upper bound on $\bar{r}$, we need to specify the trade-offs in $\epsilon_n$ and $H$ and $L$. To do so, we utilize the MLP-specific approximation rate of Lemma \ref{lem:approx-net} and the embedding of Lemma \ref{lem:embedding}. Lemma \ref{lem:embedding} implies that, for any $\epsilon_n$, one can embed the approximation class $\cF_{\rm DNN}$ given by Lemma \ref{lem:approx-net} into a standard MLP architecture $\cF_{\rm MLP}$, where specifically
\begin{align*}
	H & = H(\epsilon_n) \leq W(\epsilon_n) L(\epsilon_n) \leq C^2  \epsilon_n^{-\frac{d}{\beta}} ( \log (1/\epsilon_n) + 1)^2, \\
	L & = L(\epsilon_n) \leq C \cdot ( \log (1/\epsilon_n) + 1).
\end{align*}
For standard MLP architecture $\cF_{\rm MLP}$, 
\[
	H^2 L^2 \log (H^2 L) \leq \tilde{C} \cdot \epsilon_n^{-\frac{2d}{\beta}} ( \log (1/\epsilon_n) + 1)^7.
\]
Thus we can optimize the upper bound
\[ 
	\bar{r} \leq C' \left( \sqrt{ \frac{\epsilon_n^{-\frac{2d}{\beta}} ( \log (1/\epsilon_n) + 1)^7 }{n}\log n} + \sqrt{\frac{\log \log n + \gamma}{n}} + \epsilon_n  \right)
\]
by choosing $\epsilon_n = n^{-\frac{\beta}{2(\beta+d)}}$, $H \asymp \cdot n^{\frac{d}{2(\beta + d)}} \log^2 n$, $L \asymp \cdot \log n$. This gives
\[
	\bar{r} \leq C  \left( n^{-\frac{\beta}{2(\beta+d)}} \log^4 n + \sqrt{\frac{\log \log n + t'}{n}} \right).
\]
Hence putting everything together, with probability at least $1 - \exp(-\gamma)$,
\begin{align*}
	\E (\hat{f} - f_* )^2 \leq \bar{r}^2 \leq  C  \left( n^{-\frac{\beta}{\beta+d}} \log^8 n + \frac{\log \log n + \gamma}{n} \right) \enspace, \\
	\E_n (\hat{f} - f_* )^2  \leq (2\bar{r})^2 \leq 4C \left( n^{-\frac{\beta}{\beta+d}} \log^8 n + \frac{\log \log n + \gamma}{n} \right) \enspace.
\end{align*}
This completes the proof of Theorem \ref{thm:mlp rate}.

\subsection{Proof of Corollaries \ref{thm:optimal rate} and \ref{thm:fixed width}}
	\label{appx:coro}

For Corollary \ref{thm:optimal rate}, we want to optimize
\[
	\frac{W L \log U}{n}\log n + \frac{\log \log n + \gamma}{n} + \epsilon_{\rm DNN}^2.
\]
\citet[Theorem 1]{Yarotsky2017_NN} shows that for the approximation error $\epsilon_{\rm DNN}$ to obey $\epsilon_{\rm DNN} \leq \epsilon$, it suffices to choose
$W, U \propto \epsilon^{-\frac{d}{\beta}} (\log(1/\epsilon)+1)$ and $L \propto  (\log(1/\epsilon)+1)$, given the specific architecture described therein. Therefore, we attain $\epsilon \asymp n^{-\beta/(2\beta + d)}$ by setting $W, U \asymp n^{d/(2\beta + d)}$ and $L \asymp \log n$, yielding the desired result.

For Corollary \ref{thm:fixed width}, we need to optimize
\[
	\frac{H^2 L_2 \log (HL)}{n}\log n + \frac{\log \log n + \gamma}{n} + \epsilon_{\rm MLP}^2.
\]
\citet[Theorem 1]{Yarotsky2018_WP} shows that for the approximation error $\epsilon_{\rm MLP}$ to obey $\epsilon_{\rm MLP} \leq \epsilon$, it suffices to choose
$H \propto 2d+10$ and $L \propto \epsilon^{-\frac{d}{2}}$, given the specific architecture described therein. Thus, for $\epsilon \asymp n^{-1/(2 + d)}$ we take $L \asymp n^{-d/(4 + 2d)}$, and the result follows.

\section{Supporting Lemmas}
	\label{appx:lemmas}

First, we show that one can embed a feedforward network into the multi-layer perceptron architecture by adding auxiliary hidden nodes. This idea is due to \cite{Yarotsky2018_WP}.

\begin{lemma}[Embedding]
	\label{lem:embedding}
	For any function $f \in \cF_{\rm DNN}$, there is a $g \in \cF_{\rm MLP}$, with $H \leq WL+U$, such that $g = f$. 
\end{lemma}
\begin{proof}
	\begin{figure}[!ht] 
	  \centering
		\includegraphics[width=0.45\textwidth]{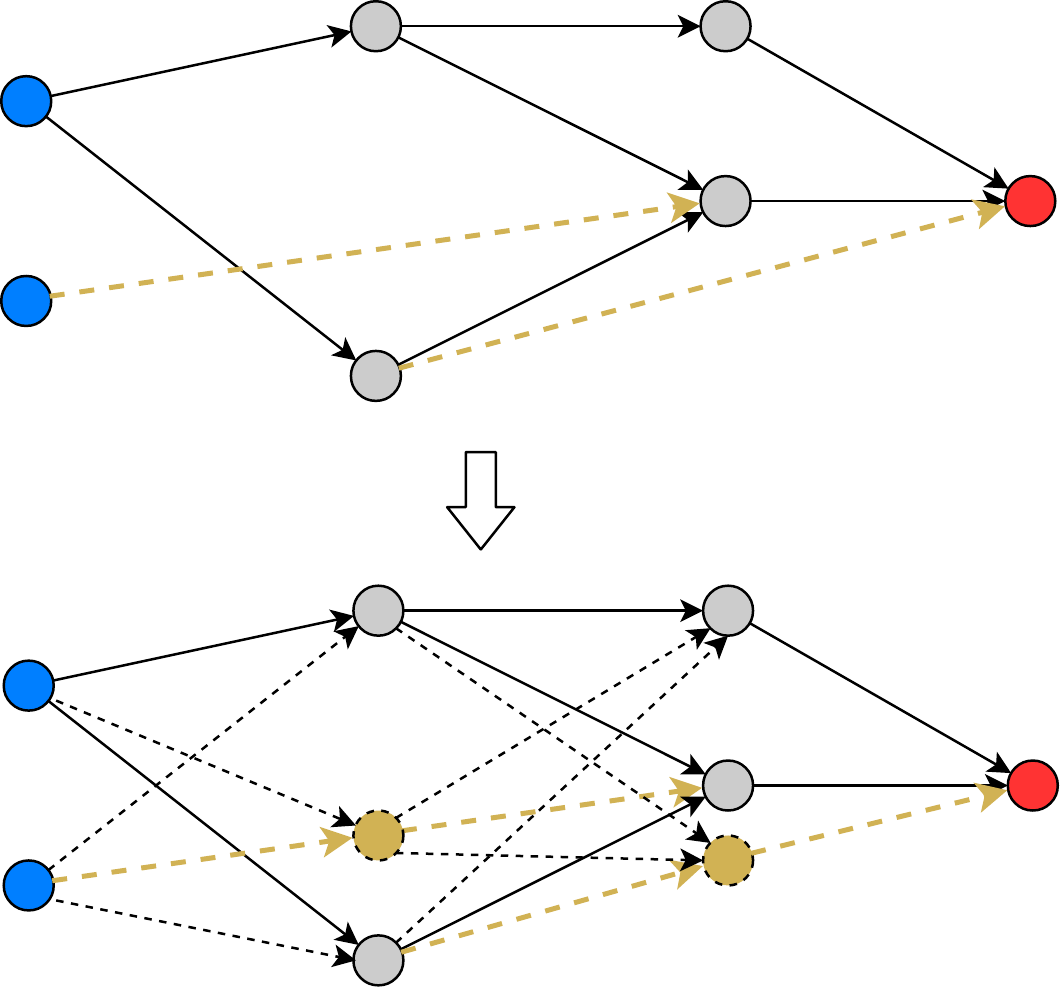}
		\caption{Illustration of how to embed a feedforward network into a multi-layer perceptron, with auxiliary hidden nodes (shown in yellow).}
		  \label{fig:DNN-embed-MLP}
	\end{figure}
	
	The idea is illustrated in Figure \ref{fig:DNN-embed-MLP}. For the edges in the directed graph of $f \in \cF_{\rm DNN}$ that connect nodes not in adjacent layers (shown in yellow in Figure \ref{fig:DNN-embed-MLP}), one can insert auxiliary hidden units in order to simply ``pass forward'' the information. 
	The number of such auxiliary ``passforward units'' is at most the number of offending edges times the depth $L$ (i.e.\ for each edge, at most $L$ auxiliary nodes are required), and this is bounded by $W L$. Therefore the width of the MLP network that subsumes the original is upper bounded by $W L + U$ while still maintaining the required embedding that for any $f_\theta \in \cF_{\rm DNN}$, there is a $g_{\theta'} \in \cF_{\rm MLP}$ such that $g_{\theta'} = f_{\theta}$. In order to match modern practice we only need to show that auxiliary units can be implemented with ReLU activation. This can be done by setting the constant (``bias'') term $b$ of each auxiliary unit large enough to ensure $\sigma(\tilde{\bx}'\bw + b) = \tilde{\bx}'\bw + b$, and then subtracting the same $b$ in the last receiving unit along the path. 
\end{proof}

Next, we give two properties of the Rademacher complexity that we require \cite[see][]{mendelson2003few}.
\begin{lemma}[Contraction]
	\label{lem:contraction}
		Let $\phi: \R \rightarrow \R$ be a Lipschitz function $|\phi(f_1) - \phi(f_2)| \leq L |f_1 - f_2|$, then
		\[
			\E_\eta R_n \{ \phi \circ f: f \in \cF \} \leq 2 L \E_\eta R_n \cF.
		\]
\end{lemma}

\begin{lemma}[Dudley's Chaining]
	\label{lem:dudley-chaining}
	Let $\cN(\delta, \cF, \| \cdot\|_n)$ denote the metric entropy for class $\cF$ (with covering radius $\delta$ and metric $\| \cdot \|_n$), then
	\[
		\E_\eta R_n \{ f: f \in \cF, \| f \|_n \leq r \} \leq \inf_{0 < \alpha < r} \left\{ 4 \alpha + \frac{12}{\sqrt{n}} \int_\alpha^r \sqrt{\log \cN(\delta, \cF, \| \cdot\|_n) } d\delta  \right\} \enspace.
	\]
	Furthermore, because $\| f \|_n \leq \max_i |f(\bx_i)|$, and therefore $\cN(\delta, \cF, \| \cdot\|_n) \leq \cN(\delta, \cF|_{\bx_1, \ldots, \bx_n}, \infty)$ and so the upper bound in the conclusions also holds with $\cN(\delta, \cF|_{\bx_1, \ldots, \bx_n}, \infty)$.
\end{lemma}


The next two results, Theorems 12.2 and 14.1 in \cite{Anthony-Bartlett1999_book}, show that the metric entropy may be bounded in terms of the pseudo-dimension and that the latter is bounded by the Vapnik-Chervonenkis (VC) dimension.
\begin{lemma}
	\label{lem:ps-dim}
	Assume for all $f \in \cF$,  $\| f \|_\infty \leq M$. Denote the pseudo-dimension of $\cF$ as ${\rm Pdim}(\cF)$, then for $n \geq {\rm Pdim}(\cF)$, we have for any $\delta$,
	\[
		 \cN(\delta, \cF|_{x_1, \ldots, x_n}, \infty) \leq \left( \frac{2e M \cdot n}{\delta \cdot {\rm Pdim}(\cF)} \right)^{{\rm Pdim}(\cF)} \enspace.
	\]
\end{lemma}


The following symmetrization lemma bounds the empirical processes term using Rademacher complexity, and is thus a crucial piece of our localization. This is a standard result based on Talagrand's concentration, but here special care is taken with the dependence on the variance. 
\begin{lemma}[Symmetrization, Theorem 2.1 in \cite{bartlett2005local}]
	\label{lem:symmetrization}
	For any $g \in \cG$, assume that $|g| \leq G$ and $\V[g] \leq V$. Then for every $\gamma > 0$, with probability at least $1 - e^{-\gamma}$
	\[
		\sup_{g \in \cG} \left\{ \E g - \E_n g \right\} \leq 3 \E R_n \cG + \sqrt{\frac{2V \gamma}{n}} + \frac{4 G}{3} \frac{\gamma}{n} \enspace ,
	\]
	and with probability at least $1 - 2e^{-t}$
	\[
		\sup_{g \in \cG} \left\{ \E g - \E_n g \right\} \leq 6 \E_\eta R_n \cG + \sqrt{\frac{2V \gamma}{n}} + \frac{23 G}{3} \frac{\gamma}{n} \enspace.
	\]
	The same result holds for $\sup_{g \in \cG} \left\{  \E_n g  - \E g \right\}$.	
\end{lemma}

When bounding the complexity of $\cF_{\rm DNN}$, we use the following result. \cite{Bartlett-etal2017_COLT} also verify these bounds for the VC-dimension.
\begin{lemma}[Theorem 6 in \cite{Bartlett-etal2017_COLT}, ReLU case]
	\label{lem:VC}
	Consider a ReLU network architecture $\cF = \cF_{\rm DNN}(W, L, U)$, then the pseudo-dimension is sandwiched as
	\[
		c \cdot WL \log (W/L) \leq {\rm Pdim}(\cF) \leq C \cdot WL \log W,
	\]
	with some universal constants $c, C>0$. 
\end{lemma}

For multi-layer perceptrons we use the following approximation result, Theorem 1 of \cite{Yarotsky2017_NN}.
\begin{lemma}
	\label{lem:approx-net}
	There exists a network class $\cF_{\rm DNN}$, with ReLU activation, such that for any $\epsilon>0$:
	\begin{enumerate}
		\item $\cF_{\rm DNN}$ approximates the $W^{\beta, \infty}([-1,1]^d)$ in the sense for any $f_* \in W^{\beta, \infty}([-1,1]^d)$, there exists a $f_n(\epsilon) \deq f_n \in \cF_{\rm DNN}$ such that
		\[
			\| f_n - f_* \|_\infty \leq \epsilon,
		\]
		\item and $\cF_{\rm DNN}$ has $L(\epsilon) \leq C \cdot ( \log (1/\epsilon) + 1)$ and $W(\epsilon), U(\epsilon) \leq C \cdot \epsilon^{-\frac{d}{\beta}} ( \log (1/\epsilon) + 1)$.
	\end{enumerate}
	Here $C$ only depends on $d$ and $\beta$.
\end{lemma}

For completeness, we verify the requirements on the loss functions, Equation \eqref{eqn:loss}, for several examples. We first treat least squares and logistic losses, in slightly more detail, as these are used in our subsequent inference results and empirical application.
\begin{lemma}
	\label{lem:ls and logit}
	Both the least squares \eqref{eqn:ols} and logistic \eqref{eqn:logit} loss functions obey the requirements of Equation \eqref{eqn:loss}. For least squares, $c_1 = c_2 = 1/2$ and $C_\ell = M$. For logistic regression, $c_1 = \left(2(\exp(M)+\exp(-M)+2)\right)^{-1}$, $c_2 = 1/8$ and $C_\ell=1$. 
\end{lemma}
\begin{proof}
	The Lipschitz conditions are trivial. For least squares, using iterated expectations
	\begin{align*}
		2 \E \ell(f, \bZ) - 2 \E \ell (f_*, \bZ) & =  \E \left[ - 2 Y f + f^2 + 2 Yf_* - f_*^2 \right]    \\
		& =  \E \left[ - 2 f_* f(\bx) + f^2 + 2 (f_*)^2 - f_*^2 \right]    \\
		&=  \E \left[ (f - f_*)^2\right].
	\end{align*}
	For logistic regression,
	\begin{align*}
		\E [\ell(f, \bZ)] - \E [\ell (f_*, \bZ)] &= \E \left[ - \frac{\exp(f_*)}{1+\exp(f_*)} (f - f_*) + \log \left( \frac{1+\exp(f)}{1+\exp(f_*)} \right) \right].
	\end{align*}
	Define $h_a(b) = - \frac{\exp(a)}{1+\exp(a)} (b-a) + \log \left( \frac{1+\exp(b)}{1+\exp(a)} \right)$, then 
	\begin{align*}
		h_a(b) = h_a(a) + h_a'(a)(b-a) + \frac{1}{2} h_a''\left(\xi a+ (1-\xi) b \right)  (b-a)^2
	\end{align*}
	and $h_a''(b) = \frac{1}{\exp(b)+\exp(-b)+2} \leq \frac{1}{4}$. The lower bound holds as $|\xi f_* + (1-\xi) f| \leq M$. 
\end{proof}

Beyond least squares and logistic regression, we give three further examples, discussed in the general language of generalized linear models. Note that in the final example we move beyond a simple scalar outcome.
\begin{lemma}
	\label{lem:glms}
For a convex function $g(\cdot):\mathbb{R} \rightarrow \mathbb{R}$, consider the generalized linear loss function $\ell(f, \bz) = -\langle y, f(\bx)\rangle + g(f(\bx))$. The curvature and the Lipschitz conditions in \eqref{eqn:loss} will hold given specific $g(\cdot)$. In each case, the loss function corresponds to the negative log likelihood function.
\begin{enumerate}
	\item Poisson: $g(t)=\exp(t)$, with $f_*(\bx) = \log \E[y|\bX = \bx]$. 
	\item Gamma: $g(t)=-\log t$, with $f_*(\bx) = -(\E[y|\bX = \bx])^{-1}$. 
	\item Multinomial Logistic, $K+1$ classes: $g(t) = \log (1+\sum_{k\in K} \exp(t^{[k]}))$, with $$\exp(f_*^{[k]}(\bx))/(1+\sum_{k'\in K} \exp(f_*^{[k']}(\bx))) = \E[y^{[k]}|\bX = \bx].$$ Here $v^{[k]}$ denotes the $k$-th coordinate of a vector $\bv$.
\end{enumerate}
\end{lemma}
\begin{proof}
	Denote $\nabla g$, ${\rm Hessian}[g]$ to be the gradient and Hessian of the convex function $g$.
	By the convexity of $g$, the optimal $f_*$ satisfies $\E[\partial \ell(f_*, \bZ) / \partial f |\bX = \bx] = 0$, which implies
	\begin{align*}
		\nabla g( f_* ) = \E[ Y | \bX = \bx].
	\end{align*} 
	If $2 c_0 \preceq {\rm Hessian}[g(f)] \preceq 2c_1$ for all $f$ of interest, then the curvature condition in \eqref{eqn:loss} holds, because
	\begin{align*}
		\E [\ell(f, \bZ)] - \E [\ell (f_*, \bZ)] & = \E[ - \langle \nabla g(f_*), f - f_* \rangle + g(f) - g(f_*) ] \\
		& = \frac{1}{2} \E\langle f-f_* , {\rm Hessian}[g(\tilde f)] f-f_* \rangle \\
		& \geq  c_0 \E \| f-f_* \|^2,
	\end{align*}
	and the parallel argument for $\leq c_1 \E \| f-f_* \|^2$. The Lipschitz condition is equivalent to $\| \nabla g(f) \| \leq C'_\ell$ for all $f$ of interest, with bounded $Y$. 
	
	For our three examples in particular, we have the following.
	\begin{enumerate}
		\item For Poisson regression:
		\begin{align*}
			\| \nabla c(f) \| = | \exp(f)| \leq \exp(M), \qquad \quad {\rm Hessian}[c(f)] = \exp(f) \in [\exp(-M), \exp(M)].
		\end{align*}
		
		\item For Gamma regression, bounding $-Y$ above and below is equivalent to $1/M \leq \| f \| \leq M$ and therefore:
		\begin{align*}
			\| \nabla c(f) \| = | 1/f| \leq M,~~{\rm Hessian}[c(f)] = 1/f^2 \in [1/M^2, M^2].  
		\end{align*} 
		
		\item For multinomial logistic regression, with general $K$, we know
	\begin{align*}
		\| \nabla c(f) \| &\leq 1 \\
		{\rm Hessian}[c(f)] &= {\rm diag}\{ z \} - z z^\top ~~\text{where}~ z^{[k]}:=\frac{\exp(f^{[k]})}{1+\sum_{k'} f^{[k']}}.
	\end{align*}
	One can easily verify that the eigenvalues are bounded in the following sense, for bounded $f$,
	\begin{align*}
		\frac{1}{(1+K\exp(M))^2} \leq \lambda({\rm Hessian}[c(f)]) \leq \frac{\exp(M)}{1+(K-1)\exp(-M)+\exp(M)}.
	\end{align*}

	\end{enumerate}
	
	This completes the proof.
\end{proof}

Our last result is to verify condition (c) of Theorem \ref{thm:ate}. We do so using our localization, which may be of future interest in second-step inference with machine learning methods.
\begin{lemma}
	\label{lem:logit}
	Let the conditions of Theorem \ref{thm:ate} hold. Then 
	\[
		\E_n\left[ (\hat{\mu}_t(\bx_i) - \mu_t(\bx_i)) \left(1 - \frac{\One\{t_i=t\} }{ \P[T=t\vert\bX = \bx_i]}\right)\right]   = o_P \left( n^{-\frac{\beta}{\beta+d}} \log^8 n + \frac{\log \log n}{n} \right)  = o_P\left(n^{-1/2}\right).
	\]
\end{lemma}

\begin{proof}
	Without loss of generality we can take $\bar{p}<1/2$. The only estimated function here is $\mu_t(\bx)$, which plays the role of $f_*$ here. For function(als) $L(\cdot)$ of the form
	\begin{align*}
		L(f) \deq (f(\bx_i) - f_*(\bx_i)) \left(1 - \frac{\One\{t_i=t\} }{ \P[T=t\vert\bX = \bx_i]}\right),
	\end{align*}
	it is true that 
	\begin{align*}
		\E [L(f)] = \E \left[  \left( f(\bX) - f_*(\bX) \right) \left( 1 - \frac{\E [ \One\{t_i=t\} \vert \bx_i] }{ \P[T=t\vert\bX = \bx_i]} \right) \right] = 0
	\end{align*}
	and
	\begin{align*}
		\V[L(f)] &\leq \left(1/\bar{p} - 1 \right)^2 \E \left[ \left( f(\bX) - f_*(\bX) \right)^2 \right] \leq \left(1/\bar{p} - 1 \right)^2 \bar{r}^2 \\
		|L(f)| &\leq \left(1/\bar{p} - 1 \right) 2M.
	\end{align*}
	
	For $\bar{r}$ defined in \eqref{eqn:r_bar},
	\begin{align*}
		18M \E R_n \{ f - f_*: f \in \cF, \|f - f_* \|_{L_2(\bx)} \leq \bar{r} \} &\leq \bar{r}^2 \\
		\E R_n \{ L(f): f \in \cF, \|f - f_* \|_{L_2(\bx)} \leq \bar{r} \}& \leq 2 \left(1/\bar{p} - 1 \right) \E R_n \{ f - f_*: f \in \cF, \|f - f_* \|_{L_2(\bx)} \leq \bar{r} \}
	\end{align*}
	where the first line is due to $\bar{r} > r_*$, and second line uses Lemma \ref{lem:contraction}.

	Then by the localization analysis and Lemma \ref{lem:symmetrization}, for all $f \in \cF, \|f - f_* \|_{L_2(\bx)} \leq \bar{r}$, $L(f)$ obeys 
	\begin{align*}
		 \E_n [L(f)] = \E_n [L(f)] - \E [L(f)] & \leq 6 C \bar{r}^2 + \bar{r} \sqrt{\frac{4 \left(1/\bar{p} - 1 \right)^2 t}{n}} + \frac{8 \left(1/\bar{p} - 1 \right) 3M}{3} \frac{t}{n}  \leq 4C \bar{r}^2 \\
		&\leq  C \cdot \left\{ n^{-\frac{\beta}{\beta+d}} \log^8 n + \frac{\log \log n}{n} \right\},\\
		\sup_{f \in \cF, \|f - f_* \|_{L_2(\bx)} \leq \bar{r}} \E_n[L(f)] &\leq C \cdot \left\{ n^{-\frac{\beta}{\beta+d}} \log^8 n + \frac{\log \log n}{n} \right\}.
	\end{align*}
	With probability at least $1 - \exp(-n^{\frac{d}{\beta+d}} \log^8 n)$, $\hat{f}_{\rm MLP}$ lies in this set of functions, and therefore
	\[
		\E_n [ L(\hat{f}_{\rm MLP}) ] = \E_n \left[ (\widehat{f}_{n,H,L}(x) - f_*(x))\left( 1 - \frac{1(T=t)}{P(T = t|\bx = x)} \right)  \right] \leq  C \cdot \left\{ n^{-\frac{\beta}{\beta+d}} \log^8 n + \frac{\log \log n}{n} \right\},
	\]
	as claimed.		
\end{proof}

\end{appendices}

\end{document}